\documentclass[journal]{IEEEtran}

\usepackage{amsmath}             
\usepackage{amsfonts}
\usepackage{amssymb}
\usepackage{amsthm}
\usepackage{subfigure} 
\usepackage{longtable}
\usepackage{setspace}
\usepackage{epsfig} 
\usepackage[footnotesize,center]{caption}
\usepackage{subfig}
\usepackage{graphicx}
\usepackage{caption}

\usepackage{color}
\usepackage{epstopdf}
\usepackage{dsfont}  		

\usepackage{etoolbox}

\usepackage{url}

\usepackage{algorithm}
\usepackage{algorithmic}

\theoremstyle{plain} 
\newtheorem{theorem}{Theorem}[]

\newtheorem{proposition}{Proposition}[]
\newtheorem{definition}{Definition}[]
\newtheorem{remark}{Remark}

\DeclareMathOperator{\Tr}{Tr}
\DeclareMathOperator*{\argmax}{arg\,max}

\usepackage{anysize}
\usepackage[
top		= 0.6in,
bottom	= 0.6in,
left		= 0.6in,
right		= 0.6in]{geometry}

\begin{document}

\title{\vspace{0mm}
\begin{spacing}{1.2}
Joint Optimization of Power and Data \\Transfer in Multiuser MIMO Systems\end{spacing}
}

\author{\vspace{0mm}\IEEEauthorblockN{Javier Rubio$^\star$, Antonio Pascual-Iserte$^\star$, Daniel P. Palomar$^\dagger$, and Andrea Goldsmith$^\ddagger$\\}
\IEEEauthorblockA{$^\star$Universitat Polit\`ecnica de Catalunya (UPC), Barcelona, Spain\\
 $^\dagger$Hong Kong University of Science and Technology (HKUST), Hong Kong
$^\ddagger$Stanford University, CA, USA\\
Emails:\{javier.rubio.lopez, antonio.pascual\}@upc.edu, palomar@ust.hk, andrea@wsl.stanford.edu\vspace{-7mm}} 
\thanks{The research leading to these results has received funding from the European Commission in the framework of the FP7 Network of Excellence in Wireless COMmunications NEWCOM\# (Grant agreement no. 318306), from the Spanish Ministry of Economy and Competitiveness (Ministerio de Econom\'ia y Competitividad) through the project TEC2011-29006-C03-02 (GRE3N-LINK-MAC), project TEC2013-41315-R (DISNET), and FPI grant BES-2012-052850, from the Catalan Government (AGAUR) through the grant 2014 SGR 60, from the Hong Kong Government through the research grant Hong Kong RGC 16207814, and from the NSF Center for Science of Information (CSoI): NSF-CCF-0939370.}
}

\maketitle
\begin{abstract}
We present an approach to solve the nonconvex optimization problem that arises when designing the transmit covariance matrices in multiuser multiple-input multiple-output (MIMO) broadcast networks implementing simultaneous wireless information and power transfer (SWIPT). The MIMO SWIPT problem is formulated as a general multi-objective optimization problem, in which data rates and harvested powers are optimized simultaneously. Two different approaches are applied to reformulate the (nonconvex) multi-objective problem. In the first approach, the transmitter can control the specific amount of power to be harvested by power transfer whereas in the second approach the transmitter can only control the proportion of power to be harvested among the different harvesting users. The computational complexity will also be different, with higher computational resources required in the first approach. In order to solve the resulting formulations, we propose to use the majorization-minimization (MM) approach. The idea behind this approach is to obtain a convex function that approximates the nonconvex objective and, then, solve a series of convex subproblems that will converge to a locally optimal solution of the general nonconvex multi-objective problem. The solution obtained from the MM approach is compared to the classical block-diagonalization (BD) strategy, typically used to solve the nonconvex multiuser MIMO network by forcing no interference among users. Simulation results show that the proposed approach improves over the BD approach both the system sum rate and the power harvested by users. Additionally, the computational times needed for convergence of the proposed methods are much lower than the ones required for classical gradient-based approaches. 
\end{abstract}

\section{Introduction}
\IEEEPARstart{S}{imultaneous} wireless information and power transfer (SWIPT) is a technique by which a transmitter actively feeds a receiver (or a set of receivers) with power that is sent through radio frequency (RF) signals and, simultaneously, sends useful information to the same or different receivers \cite{lu:15}. By harvesting this transmitted energy, battery-constrained mobile terminals are able to recharge their batteries and, thus, prolong their operation time \cite{paradiso:05}. Although there are many different harvesting techniques used to power devices, such as solar or wind, SWIPT technology provides an appealing solution since the transmitter can control the amount of energy that the mobile terminals need to keep alive. Historically, due to the high attenuation of microwave signals over distance, SWIPT techniques were only considered in low-power devices, such as RFID tags \cite{bi:15}. Nevertheless, recent advances in antenna technologies and RF harvesting circuits have enabled energy to be transferred and harvested much more efficiently \cite{lu:15, bi:15}.

The concept of SWIPT was first studied from a theoretical point of view by Varshney \cite{varshney:08}. He showed that, for the single-antenna additive white Gaussian noise (AWGN) channel, there exists a nontrivial trade-off in maximizing the data rate versus the power transmission. In \cite{zhang:13b}, the authors considered a multiple-input multiple-output (MIMO) scenario with one transmitter capable of transmitting information and power simultaneously to one receiver. Later, in \cite{rub:13d}, the authors extended the work in \cite{zhang:13b} by considering that multiple users were present in the broadcast MIMO system. However, the multi-stream transmit covariance optimization that arises in broadcast MIMO systems is a very difficult nonconvex optimization problem. In order to overcome that difficulty, authors in \cite{rub:13d} considered a block-diagonalization (BD) strategy \cite{spencer:04}, in which interference is pre-canceled at the transmitter. The BD technique allows for a simple solution but wastes some degrees of freedom and, thus, degrades the overall performance. Works \cite{park:13} and \cite{park:14} considered a MIMO network consisting of multiple transmitter-receiver pairs with co-channel interference. The study in \cite{park:13} focused on the case with two transmitter-receiver pairs whereas in \cite{park:14}, the authors generalized \cite{park:13} by considering that $k$ transmitter-receivers pairs were present. The work in \cite{zong:16} considered a MIMO system with single-stream transmission. In contrast to previous works where the system rate was optimized, the objective was to minimize the overall power consumption with per-user signal to interference and noise ratio (SINR) constraints and harvesting constraints. The design of multiuser broadcast networks under the framework of multiple-input single-output (MISO) beamformimg optimization has also been addressed in works such as \cite{xu:14} and \cite{shi:14}. 

There exist two approaches in the literature that deal with the nonconvex optimization of the transmit covariance matrices in multiuser multi-stream MIMO networks. The first is based on the duality principle \cite{vishwanath:03}. In \cite{gui:15}, authors applied that principle to obtain the beamforming optimization solution for the multiuser MIMO SWIPT broadcast channel. However, that work considered an overall (sum) harvesting constraint instead of individual per-user harvesting constraints. The second approach is based on the minimization of the mean square error (MSE) \cite{christensen:08}. However, this technique cannot be applied to the SWIPT framework due to fact that the resulting problem remains nonconvex.  

The main difference of our work with respect to the previous works described above is that we assume a broadcast multiuser multi-stream (non BD-based) MIMO SWIPT network, in which (per-user) harvested power and information transfer must be optimized simultaneously. We model our transmitter design as a multi-objective problem in which the scenarios studied in \cite{zhang:13b} and \cite{rub:13d} are shown to be particular solutions of the proposed framework. Additionally, we assume that interference is not pre-canceled (i.e., the BD approach is not applied) and, thus, both larger information transfer and harvested power can be achieved simultaneously. The resulting problem is nonconvex and very difficult to solve. In order to obtain local solutions, we derive different methods based on majorization-minimization (MM) techniques. By means of this strategy, we are able to reformulate our original nonconvex problem into a series of convex subproblems that are easily solved (i.e., through algorithms that have a very low computational complexity) and whose solutions converge to a locally optimal solution of the original nonconvex problem. 

The remainder of this paper is organized as follows. In Section \ref{sec_math_pre}, we introduce a summary of the mathematical techniques employed in this paper. In Section \ref{sec_sysm} we present the system and signal models and the problem formulation. In Section \ref{sec_mm_app} we derive the mathematical modeling required to reformulate the original nonconvex problem into convex subproblems that are solved using the MM approach. In Section \ref{sec_num_mm}, we evaluate the performance of the proposed methods and, finally, in Section \ref{sec_con_mm}, we draw some conclusions.
 
\emph{Notation:} We adopt the notation of using boldface lower case for vectors ￼$\textbf{x}$ and upper case for matrices $\textbf{X}$. The transpose, conjugate transpose (hermitian), and inverse operators are denoted by the superscripts $(\cdot)^T$￼, $(\cdot)^H$￼￼￼￼, and $(\cdot)^{-1}$￼, respectively. $\Tr(\cdot)$ and $\det(\cdot)$ denote the trace and the determinant of a matrix, respectively. $\text{vec}(\textbf{X})$ is a column vector resulting from stacking all columns of $\textbf{X}$. We use $\textbf{X}$ to denote the $N-$tuple $\textbf{X} \triangleq (\textbf{X}_i)^N_{i=1} = (\textbf{X}_1,\dots,\textbf{X}_N)$ and $||\cdot||_F$ to denote the matrix Frobenius norm.
\section{Mathematical Preliminaries}
\label{sec_math_pre}
\subsection{Multi-Objective Optimization}
\label{sec-mop}
Multi-objective optimization (also known as multi-criteria optimization or vector optimization) is a type of optimization that involves multiple objective functions that are optimized simultaneously \cite{ehrgott}. For a nontrivial multi-objective problem, in general, there does not exist a single solution that simultaneously optimizes each objective. In that case, the objective functions are said to be conflicting, and there exists a (possibly infinite) number of Pareto optimal solutions. A solution is called Pareto optimal if none of the objective functions can be improved in value without degrading some of the other objective values.

\noindent\emph{1) Definitions}

\begin{definition}[\cite{ehrgott}]A multi-objective problem can be formally expressed as
\begin{alignat}{2}
&\mathop{\text{maximize}}_{\textbf{x}} \quad  \textbf{f}(\textbf{x}) = (f_1(\textbf{x}), \dots, f_K(\textbf{x})) \label{op:mop1} \\
&\textrm{subject to} \quad \textbf{x} \in \mathcal{X},\nonumber
\end{alignat}
where $f_k : \mathbb{C}^N \rightarrow \mathbb{R}$ for $k=1, \dots, K$ and $\mathcal{X}$ is the feasible set that represents the constraints. Let $\mathcal{Y}$ be the set of all attainable points for all feasible solutions, i.e., $\mathcal{Y} = \textbf{f}(\mathcal{X})$.
\end{definition}
\noindent\emph{2) Efficient Solutions}

\begin{definition}[\cite{ehrgott}, Definition 2.1]
A point $\textbf{x}\in\mathcal{X}$ is called Pareto optimal if there is no other $\textbf{x}'\in\mathcal{X}$ such that \emph{$\textbf{f}(\textbf{x}') \succeq \textbf{f}(\textbf{x})$}, where $\succeq$ refers to the component-wise inequality, i.e., $f_i(\textbf{x}')\ge f_i(\textbf{x})$, $i=1,\dots,K$.
\end{definition}

Sometimes, ensuring Pareto optimality for some problems is difficult. Due to this, the condition of optimality can be relaxed as follows.

\begin{definition}[\cite{ehrgott}, Definition 2.24]
A point $\textbf{x}\in\mathcal{X}$ is called weakly Pareto optimal (or weakly efficient) if there is no other $\textbf{x}'\in\mathcal{X}$ such that \emph{$\textbf{f}(\textbf{x}') \succ \textbf{f}(\textbf{x})$}, where $\succ$ refers to the strict component-wise inequality, i.e., $f_i(\textbf{x}')> f_i(\textbf{x})$, $i=1,\dots,K$. All Pareto optimal solutions are also weakly Pareto optimal.

\end{definition}

\noindent\emph{3) Finding Pareto Optimal Points}

There are several methods for finding the Pareto points of a multi-objective problem. In the sequel, we present three different (scalarization) techniques.

\noindent\emph{3.1) Weighted sum method:} the simplest scalarization technique is the weighted sum method which collapses the vector-objective into a single-objective component sum:
\begin{alignat}{2}
&\mathop{\text{maximize}}_{\textbf{x}\in \mathcal{X}} \quad  \sum_{k=1}^K \beta_kf_k(\textbf{x}), \label{op:mop2} 
\end{alignat}
where $\beta_k$ are real non-negative weights. The following results present the relation between the optimal solutions of \eqref{op:mop2} and the Pareto optimal points of the original problem \eqref{op:mop1}.

\begin{proposition}[\cite{ehrgott}, Proposition 3.9]
Suppose that ${\textbf{x}^\star}$ is an optimal solution of \eqref{op:mop2}. Then, ${\textbf{x}}^\star$ is weakly efficient.
\end{proposition}

\begin{proposition}[\cite{ehrgott}, Proposition 3.10]
Let $\mathcal{X}$ be a convex set, and let $f_k$ be concave functions, $k=1,\dots,K$. If ${\textbf{x}}^\star$ is weakly efficient, there are some $\beta_k\ge0$ such that ${\textbf{x}}^\star$ is an optimal solution of \eqref{op:mop2}.
\end{proposition}

As as result, convexity is apparently required for finding all weakly Pareto optimal points with the weighted sum method, which means that if the original problem is not convex, all the Pareto optimal points may not be found by using the weighted sum method. However, there are other weighted sum techniques in the literature (see, for example, the adaptive weighted sum method \cite{kim:06}) that are able to find all Pareto optimal points for nonconvex problems at the expense of a higher computational complexity. 

\noindent\emph{3.2) Epsilon-constraint method:} in this method, only one of the original objectives is maximized while the others are transformed into constraints:
\begin{alignat}{2}
&\mathop{\text{maximize}}_{\textbf{x}\in \mathcal{X}} \quad  f_j(\textbf{x}) \label{op:mop3} \\
&\textrm{subject to} \quad f_k(\textbf{x}) \ge \epsilon_k, \quad k=1,\dots,K, \,\, k\neq j.\nonumber
\end{alignat}
Let us introduce the following results.

\begin{proposition}[\cite{ehrgott}, Proposition 4.3]
Let ${\textbf{x}}^\star$ be an optimal solution of \eqref{op:mop3} for some $j$. Then ${\textbf{x}}^\star$ is weakly Pareto optimal.
\end{proposition}

\begin{proposition}[\cite{ehrgott}, Proposition 4.5]
A feasible solution ${\textbf{x}}^\star\in\mathcal{X}$ is Pareto optimal if, and only if, there exists a set of $\hat{\epsilon}_k, k=1,\dots,K$ such that ${\textbf{x}}^\star$ is an optimal solution of \eqref{op:mop3} for all $j=1,\dots,K$.
\end{proposition}

Contrarily to the weighted sum method, convexity is not needed in the previous two propositions (but convexity is still typically required to solve problems like \eqref{op:mop3}).

\noindent\emph{3.3) Hybrid method:} this method combines the previous two methods, i.e., the weighted sum method and the epsilon-constraint method. In this case, the scalarized problem to be solved has a weighted sum objective and constraints on all (or some) objectives. 
\begin{alignat}{2}
&\mathop{\text{maximize}}_{\textbf{x}\in \mathcal{X}} \quad  \sum_{k\in\mathcal{K}_1} \beta_kf_k(\textbf{x}) \label{op:mop4} \\
&\textrm{subject to} \quad f_k(\textbf{x}) \ge \epsilon_k, \quad k\in\mathcal{K}_2,\nonumber
\end{alignat}
where $|\mathcal{K}_1|\le K$ and $|\mathcal{K}_2|\le K$, being $|\mathcal{A}|$ the cardinality of set $\mathcal{A}$, and $\beta_k$ are real non-negative weights. 

\subsection{Majorization-Minimization Method}
\label{sec_mm}
The MM is an approach to solve optimization problems that are too difficult to solve in their original formulation. The principle behind the MM method is to transform a difficult problem into a sequence of simple problems. Interested readers may refer to \cite{hunter:04} and references therein for more details.

The method works as follows. Suppose that we want to maximize $f_0(\textbf{x})$ over $\mathcal{X}$. In the MM approach, instead of maximizing the cost function $f_0(\textbf{x})$ directly, the algorithm optimizes a sequence of approximate objective functions that minorize $f_0(\textbf{x})$, producing a sequence $\{\textbf{x}^{(k)}\}$ according to the following update rule:
\begin{equation}
\textbf{x}^{(k+1)} = \argmax_{\textbf{x}\in\mathcal{X}} \,\, \hat{f}_0(\textbf{x}, \textbf{x}^{(k)}),\label{surro_problem}
\end{equation}
where $\textbf{x}^{(k)}$ is the point generated by the algorithm at iteration $k$ and $\hat{f}_0(\textbf{x}, \textbf{x}^{(k)})$ known as surrogate function is the minorization function of $f_0(\textbf{x})$ at $\textbf{x}^{(k)}$, i.e., it has to be a global lower bound tight at $\textbf{x}^{(k)}$. Problem \eqref{surro_problem} will be referred as surrogate problem. In addition, the surrogate function must also be continuous in $\textbf{x}$ and $\textbf{x}^{(k)}$. The last condition that the surrogate function must fulfill is that the directional derivatives\footnote{Let $f : \mathbb{C}^N\rightarrow \mathbb{R}$. Then, the directional derivative of $f(\textbf{x})$ in the direction of vector $\textbf{d}$ is given by $f'(\textbf{x};\textbf{d}) \triangleq \lim_{\lambda\rightarrow 0} \frac{f(\textbf{x} + \lambda\textbf{d}) - f(\textbf{x})}{\lambda}$.} of itself and of the original objective function $f_0(\textbf{x})$ must be equal at the point $\textbf{x}^{(k)}$. All in all, the four conditions are as follows:
\begin{eqnarray}
\hspace{-5mm}(\text{A}1):&& \quad \hat{f}_0(\textbf{x}^{(k)}, \textbf{x}^{(k)}) = f_0(\textbf{x}^{(k)}), \quad \forall \textbf{x}^{(k)}\in\mathcal{X},\label{cond1}\\
\hspace{-5mm}(\text{A}2):&& \quad \hat{f}_0(\textbf{x}, \textbf{x}^{(k)}) \le f_0(\textbf{x}), \quad \forall \textbf{x}, \textbf{x}^{(k)}\in\mathcal{X},\\
\hspace{-5mm}(\text{A}3):&& \quad \hat{f}'_0(\textbf{x},\textbf{x}^{(k)};\textbf{d})|_{\textbf{x}=\textbf{x}^{(k)}} = f'_0(\textbf{x}^{(k)};\textbf{d}),  \nonumber\\
&&\quad \forall \textbf{d} \text{ with } \textbf{x}^{(k)} + \textbf{d}\in\mathcal{X},\\
\hspace{-5mm}(\text{A}4):&& \quad \hat{f}_0(\textbf{x}, \textbf{x}^{(k)}) \text{ is continuous in } \textbf{x} \text{ and } \textbf{x}^{(k)}.\label{cond4}
\end{eqnarray}
Under assumptions $(\text{A}1)-(\text{A}4)$, every limit point of the sequence $\{\textbf{x}^{(k)}\}$ is a locally optimal point of the original problem (globally optimal if the problem is convex) (see \cite{hunter:04} for details).

\section{Problem Formulation}
\label{sec_sysm}

Let us consider a wireless broadcast multiuser system consisting of one base station (BS) transmitter equipped with $n_T$ antennas and a set of $K$ receivers, denoted as $\mathcal{U}_T = \{1, 2, \dots , K\}$, where the $k$-th receiver is equipped with $n_{R_k}$ antennas. We assume that a given user is not able to decode information and to harvest energy simultaneously, and that a user being served with information by the BS uses all the energy to decode the signal. Thus, the set of users is partitioned into two disjoint subsets. One that contains the information users, denoted as $\mathcal{U}_I \subseteq \mathcal{U}_T$ with $|\mathcal{U}_I| = N$, and the other subset that contains harvesting users, denoted as  $\mathcal{U}_E \subseteq \mathcal{U}_T$ with $|\mathcal{U}_E| = M$. Therefore, $\mathcal{U}_I \,\cap\, \mathcal{U}_E = \emptyset$ and  $|\mathcal{U}_I| + |\mathcal{U}_E| = N+M = K$.\footnote{In this paper, we assume for simplicity in the formulation that a user belongs to either the harvesting set or the information set and that both sets are known and fixed. This assumption could be generalized by considering that some users are not selected in either set as well as by defining which particular users are scheduled in each particular set (i.e., user grouping strategies). However, this falls out of the scope of this paper.} Without loss of generality (w.l.o.g.), let us index users as $\mathcal{U}_I = \{1,\dots,N\}$ and $\mathcal{U}_E=\{N+1,\dots,N+M\}$.

The equivalent baseband channel from the BS to the $k$-th receiver is denoted by $\textbf{H}_k \in \mathbb{C}^{n_{R_k}\times n_{T}}$. It is also assumed that the set of matrices $\{\textbf{H}_k\}$ is known to the BS and to the corresponding receivers (the case of imperfect CSI is out of the scope of the paper).

As far as the signal model is concerned, the received signal for the $i$-th information receiver can be modeled as
\begin{equation}
\textbf{y}_i = \textbf{H}_i\textbf{B}_i\textbf{x}_i + \textbf{H}_i\sum_{\substack{k\in\mathcal{U}_I\\k\neq i}}\textbf{B}_k\textbf{x}_k + \textbf{n}_i, \quad \forall i \in \mathcal{U}_I.
\label{sig_model}
\end{equation}
In the previous notation, $\textbf{B}_i\textbf{x}_i$ represents the transmitted signal for user $i \in \mathcal{U}_I$, where $\textbf{B}_i\in \mathbb{C}^{n_{T}\times n_{S_i}}$ is the precoder matrix and $\textbf{x}_i\in \mathbb{C}^{n_{S_i}\times 1}$ represents the information symbol vector. It is also assumed that the signals transmitted to different users are independent and zero mean. $n_{S_i}$ denotes the number of streams assigned to user $i \in \mathcal{U}_I$ and we assume that $n_{S_i} = \min\{n_{R_i}, n_T\} \, \forall i \in \mathcal{U}_I$. The transmit covariance matrix is $\textbf{S}_i = \textbf{B}_i\textbf{B}_i^H$ if we assume w.l.o.g. that $\mathbb{E}\left[\textbf{x}_i\textbf{x}_i^H\right] = \textbf{I}_{n_{S_i}}$. $\textbf{n}_i\in \mathbb{C}^{n_{R_i}\times 1}$ denotes the receiver noise vector, which is considered Gaussian with $\mathbb{E}\left[\textbf{n}_i\textbf{n}_i^H\right] = \textbf{I}_{n_{R_i}}$\footnote{We assume that noise power $\sigma^2 = 1$ w.l.o.g., otherwise we could simply apply a scale factor at the receiver and re-scale the channels accordingly.}. Note that the middle term of \eqref{sig_model} is an interference term. The covariance matrix of the interference plus noise is written as
\begin{equation}\label{interf}
\boldsymbol\Omega_i(\textbf{S}_{-i}) = \textbf{H}_i\textbf{S}_{-i}\textbf{H}^H_i + \textbf{I}, \quad \forall i \in \mathcal{U}_I,
\end{equation}
where $\textbf{S}_{-i} = \sum_{\substack{k\in\mathcal{U}_I\\k\neq i}}\textbf{S}_k$. Let $\tilde{\textbf{x}} = \textbf{B}\textbf{x}$ denote the signal vector transmitted by the BS, where the joint precoding matrix is defined as $\textbf{B} = [\textbf{B}_1  \quad\dots \quad \textbf{B}_N] \in \mathbb{C}^{n_T \times n_S}$, being $n_S = \sum_{i\in\mathcal{U}_I}n_{S_i}$ the total number of streams of all information users, and the data vector as $\textbf{x} = \left[\textbf{x}_1^T  \quad\dots \quad \textbf{x}_N^T\right]^T \in \mathbb{C}^{n_S \times 1}$, that must satisfy the power constraint formulated as $\mathbb{E}[\|{\tilde{\textbf{x}}}\|^2] = \sum_{i\in\mathcal{U}_I}\Tr(\textbf{S}_i)\le P_{T}$, where $P_{T}$ represents the total available transmission power at the BS.

The total RF-band power harvested by the $j$-th user from all receiving antennas, denoted by $\bar{Q}_j$, is proportional to that of the equivalent baseband signal, i.e., $\forall j \in \mathcal{U}_E,$ we have:
\begin{eqnarray}
\hspace{-5mm}\bar{Q}_j &=& \zeta_j  \mathbb{E}\Big[\Big\| \textbf{H}_j \sum_{i\in \mathcal{U}_I}\textbf{B}_i \textbf{x}_i\Big\|^2\Big] =  \zeta_j \sum_{i\in \mathcal{U}_I} \mathbb{E}[\| \textbf{H}_j \textbf{B}_i \textbf{x}_i\|^2], 
\label{eqEn}
\end{eqnarray}
where $\zeta_j$ is a constant that accounts for the loss for converting the harvested RF power to electrical power. Notice that, for simplicity, in \eqref{eqEn} we have omitted the harvested power due to the noise term since it can be assumed negligible.

The transmitter design that we propose in this paper is modeled as a nonconvex multi-objective optimization problem. The goal is to maximize, simultaneously, the individual data rates and the harvested powers of the information and harvesting users, respectively. Given this and the previous system model, the optimization problem is written as
\begin{alignat}{2}
\mathop{\text{maximize}}_{\{\textbf{S}_i\}}& \quad  \Big((R_n(\textbf{S}))_{n\in\mathcal{U}_I}, (E_m(\textbf{S}))_{m\in\mathcal{U}_E}\Big)\label{op:mob} \\
\textrm{subject to} 
&  \quad C1:\sum_{i\in\mathcal{U}_I}\Tr(\textbf{S}_i)\le P_T\nonumber\\
&  \quad C2:  \textbf{S}_i \succeq 0, \quad \forall i \in \mathcal{U}_I, \nonumber
\end{alignat}
where $\textbf{S} \triangleq (\textbf{S}_i)_{\forall i\in\mathcal{U}_I}$, the data rate expression is given by 
\begin{eqnarray}
\hspace{-5mm}R_n(\textbf{S}) &=& \log \det \left(\textbf{I} + \textbf{H}_n\textbf{S}_n\textbf{H}_n^H\boldsymbol\Omega_n^{-1}(\textbf{S}_{-n})\right)\\
&=& \log \det \left(\boldsymbol\Omega_n(\textbf{S}_{-n}) + \textbf{H}_n\textbf{S}_n\textbf{H}_n^H\right) \nonumber\\
&&- \log\det\left(\boldsymbol\Omega_n(\textbf{S}_{-n}) \right)\\
&=&  \underbrace{\log \det \left(\textbf{I} + \textbf{H}_n\bar{\textbf{S}}\textbf{H}_n^H\right)}_{\triangleq \,\,s_n(\textbf{S})} - \underbrace{\log\det\left(\boldsymbol\Omega_n(\textbf{S}_{-n}) \right)}_{\triangleq \,\,g_n(\boldsymbol\Omega_n(\textbf{S}_{-n}))},
\end{eqnarray}
with $\bar{\textbf{S}} = \sum_{k\in\mathcal{U}_I} \textbf{S}_k$, and the harvested power is given by
\begin{equation}
E_m(\textbf{S}) = \sum_{i\in\mathcal{U}_I}\Tr(\textbf{H}_m\textbf{S}_i\textbf{H}_m^H).
\end{equation}
The previous problem in \eqref{op:mob} is not convex due the objective functions (in fact, due to $\boldsymbol\Omega_i(\textbf{S}_{-i}) $) and is difficult to solve. In order to find Pareto optimal points, we can reformulate it by using any of the techniques presented in Section \ref{sec-mop}. In the following, we propose two approaches based on the weighted sum method and on the hybrid method. For convenience, we start with the hybrid method as it is the one that has received the most attention in the literature \cite{zhang:13b}, \cite{zhang:13d}. However in that literature, the interference in \eqref{interf} is assumed to be removed by the transmission strategy. This assumption makes the problem convex and hence easier to solve.

\vspace{0mm}
\subsection{Hybrid-Based Formulation to Solve \eqref{op:mob}}
\label{subsec_hyb}
In the hybrid approach, some of the objective functions are collapsed into a single objective by means of scalarization and some of the objective functions are added as constraints. In particular, the data rates are left in the objective whereas the harvesting constraints are included as individual harvesting constraints. With this particular formulation, we are able to guarantee a minimum value for the power to be harvested by the harvesting users. Thus, problem \eqref{op:mob} is formulated as

\begin{alignat}{2}
\mathop{\text{max}}_{\{\textbf{S}_i\}}& \quad \sum_{i\in\mathcal{U}_I} \omega_i\log \det \left(\textbf{I} + \textbf{H}_i\bar{\textbf{S}}\textbf{H}_i^H\right) - \omega_i\log\det\left(\boldsymbol\Omega_i(\textbf{S}_{-i}) \right)  \nonumber\\
\textrm{s. t.} 
&  \quad C1:\sum_{i\in\mathcal{U}_I}\Tr(\textbf{H}_j\textbf{S}_i\textbf{H}_j^H)  \ge Q_j, \quad \forall j \in \mathcal{U}_E  \label{op:wet1}\\
&  \quad C2:\sum_{i\in\mathcal{U}_I}\Tr(\textbf{S}_i)\le P_T\nonumber\\
&  \quad C3:  \textbf{S}_i \succeq 0, \quad \forall i \in \mathcal{U}_I, \nonumber
\end{alignat}
where $Q_j = \frac{{\bar{Q}^{\min}}_j}{\zeta_j}$, being $\{\bar{Q}^{\min}_j\}$ the set of minimum power harvesting constraints, and $\omega_i$ are some real non-negative weights. For simplicity in the notation, let us define the feasible set $\mathcal{S}_1$ as 
\begin{eqnarray}
\hspace{-5mm}\mathcal{S}_1\triangleq &\Bigg\{&\textbf{S} : \sum_{i\in\mathcal{U}_I}\Tr(\textbf{H}_j\textbf{S}_i\textbf{H}_j^H)  \ge Q_j, \,\forall j \in\mathcal{U}_E,\nonumber\\
&&\sum_{i\in\mathcal{U}_I}\Tr(\textbf{S}_i)\le P_T,\,\textbf{S}_i \succeq 0, \forall i \in\mathcal{U}_I\Bigg\}.
\end{eqnarray}
For a set of fixed harvesting constraints, the convex hull of the rate region can be obtained by varying the values of $\omega_i$. In addition, we can use the values of the weights to assign priorities to some users if user scheduling is to be implemented, following, for example, the proportional fair criterion \cite{liu:10}, \cite{andrews:01}. Notice that constraint $C1$ is associated with the minimum power to be harvested for a given user. Note also the similarities of problem \eqref{op:wet1} with the single user case presented in \cite{zhang:13b} and its extension to the multiuser case presented in \cite{rub:13d}. As commented before, the novelty is that we do not force the transmitter to cancel the interference generated among the information users (as opposed to BD approaches \cite{spencer:04}) and, thus, we allow the system to have more degrees of freedom to improve the system throughput and the harvested power simultaneously. Later in Section \ref{subsec_q1}, we will present a method based on MM to solve the nonconvex problem in \eqref{op:wet1}.  

\subsection{Weighted Sum-Based Formulation to Solve \eqref{op:mob}}
\label{subsec_we}
In situations where the exact amount of power to be harvested by harvesting users is not needed, we can also obtain Pareto optimal points by means of the simpler weighted-sum method. In this case, we can assign priorities so that some users tend to harvest more power than others, although the exact amounts cannot be controlled. As we will see later, the overall problem based on this new formulation is much easier to solve. The transmitter design is obtained through the following nonconvex optimization problem:
\begin{alignat}{2}
\mathop{\text{max}}_{\{\textbf{S}_i\}}& \quad \sum_{i\in\mathcal{U}_I}  \omega_i\log \det \left(\textbf{I} + \textbf{H}_i\bar{\textbf{S}}\textbf{H}_i^H\right) - \omega_i\log\det\left(\boldsymbol\Omega_i(\textbf{S}_{-i}) \right) \nonumber\\
&\quad + \sum_{j\in\mathcal{U}_E}\sum_{i\in\mathcal{U}_I}\alpha_j\Tr(\textbf{H}_j\textbf{S}_i\textbf{H}_j^H)\label{op:wetre} \\
\textrm{s. t.} 
&  \quad C1:\sum_{i\in\mathcal{U}_I}\Tr(\textbf{S}_i)\le P_T\nonumber\\
&  \quad C2:  \textbf{S}_i \succeq 0, \quad \forall i \in \mathcal{U}_I, \nonumber
\end{alignat}
where $\alpha_j$ are some real non-negative weights. For simplicity in the notation, let us define the feasible set $\mathcal{S}_2$ as 
\begin{equation}
\mathcal{S}_2\triangleq \Bigg\{\textbf{S} : \sum_{i\in\mathcal{U}_I}\Tr(\textbf{S}_i)\le P_T,\,\textbf{S}_i \succeq 0, \forall i \in\mathcal{U}_I\Bigg\}.
\end{equation}
As we will show later in Section \ref{subsec_q2}, the algorithm to solve \eqref{op:wetre} is easier than the algorithm to solve \eqref{op:wet1}. Hence, there is a trade-off in terms of speed of convergence of the algorithms and in terms of the harvested power control since, as we introduced before, in \eqref{op:wet1} the transmitter can fully control the amount of power to be harvested by the users whereas in \eqref{op:wetre} the transmitter can only control the proportion of the power to be harvested among the users. 

\section{MM-based Techniques to Solve Problem \eqref{op:mob}}
In this section, we present a method based on the MM philosophy to solve problems \eqref{op:wet1} and \eqref{op:wetre}. Since the original problems \eqref{op:wet1} and \eqref{op:wetre} are nonconvex, we reformulate them and make them convex before applying the MM method. This reformulation will follow two steps. In the first step, problems \eqref{op:wet1} and \eqref{op:wetre} will be convexified by using a linear approximation of the nonconvex terms. This is the approach taken in papers such as \cite{Hong:16}, \cite{scutari:14}, and \cite{You:14}. Instead of solving the reformulated (convex) problem, in the second step, we design a quadratic approximation of the remaining convex terms in order to find a surrogate problem easier to solve. Finally, we apply the MM method to the quadratic reformulation.

As benchmarks for comparison, we will consider the case of just convexifying the nonconvex terms, which is an approach taken in the previous literature, and also consider a gradient method applied directly to the nonconvex problems \eqref{op:wet1} and \eqref{op:wetre}.

Although the mathematical developments of the proposed MM approaches are more tedious than the approaches usually taken in the literature, the resulting algorithms are faster. 

\label{sec_mm_app}
\subsection{Approach to Solve the Hybrid Formulation in \eqref{op:wet1}}
\label{subsec_q1}
As we introduced before, we need to reformulate the original nonconvex problem \eqref{op:wet1} and make it convex. This will be done in two steps. Motivated by the work in \cite{scutari:14}, in this first step, we derive a linear approximation for the nonconcave (right-hand side) part of the objective function of \eqref{op:wet1}, i.e., $f_0(\textbf{S}) = \sum_{i\in\mathcal{U}_I} \omega_is_i(\textbf{S}) - \omega_ig_i(\boldsymbol\Omega_i(\textbf{S}_{-i}))$, in such a way that the modified problem is convex\footnote{In fact, by applying the approximation, the overall objective function becomes concave.}. In order to find a concave lower bound of $f_0(\textbf{S})$, $g_i(\cdot)$ can be upper bounded linearly at point $\boldsymbol\Omega^{(0)}_i = \sum_{\substack{k\in\mathcal{U}_I\\k\neq i}}\textbf{H}_i\textbf{S}^{(0)}_k\textbf{H}^H_i + \textbf{I}$ as
\begin{eqnarray}
g_i(\boldsymbol\Omega_i\hspace{-4mm}&(\hspace{-4mm}&\textbf{S}_{-i})) \le \nonumber\\
&& g_i\left(\boldsymbol\Omega^{(0)}_i\right) + \Tr\left(\left(\boldsymbol\Omega_i^{(0)}\right)^{-1}\left(\boldsymbol\Omega_i(\textbf{S}_{-i}) - \boldsymbol\Omega^{(0)}_i\right)\right) \nonumber\\
&=& \text{constant} + \Tr\left(\left(\boldsymbol\Omega_i^{(0)}\right)^{-1}\boldsymbol\Omega_i(\textbf{S}_{-i})\right) \nonumber\\
&\triangleq& \hat{g}_i(\boldsymbol\Omega_i(\textbf{S}_{-i}),\boldsymbol\Omega_i^{(0)}).\label{taylor}
\end{eqnarray}
Even though problem \eqref{op:wet1} reformulated with the previous upper bound $\hat{g}_i(\boldsymbol\Omega_i(\textbf{S}_{-i}),\boldsymbol\Omega_i^{(0)})$ is convex, we want to go one step further and apply a quadratic lower bound for the left hand side of $f_0(\textbf{S})$, i.e., $s_i(\textbf{S})$ in a way that the overall lower bound fulfills conditions $(\text{A}1)-(\text{A}4)$  presented before in Section \ref{sec_mm} and the MM method can be invoked. Note that the upper bound $\hat{g}_i(\boldsymbol\Omega_i(\textbf{S}_{-i}),\boldsymbol\Omega_i^{(0)})$ already fulfills the four conditions $(\text{A}1) - (\text{A}4)$. The idea of implementing this quadratic bound is to find a surrogate problem that is much simpler and easier to solve than the one obtained by just considering the linear bound $\hat{g}_i(\boldsymbol\Omega_i(\textbf{S}_{-i}),\boldsymbol\Omega_i^{(0)})$. \footnote{The surrogate problem obtained by just applying the bound $\hat{g}_i(\boldsymbol\Omega_i(\textbf{S}_{-i}),\boldsymbol\Omega_i^{(0)})$ will be used as benchmark. The specific mathematical details of the optimization problem and the algorithm will be described in App. \ref{app_bench}.}

We now focus attention on deriving the surrogate function for the left hand side of $f_0(\textbf{S})$, i.e., $s_i(\textbf{S})$. In order for the surrogate problem to be easily solved, we force the surrogate function of $s_i(\textbf{S})$ around $\bar{\textbf{S}}^{(0)}$ to be quadratic, where $\bar{\textbf{S}}^{(0)} = \sum_{k\in\mathcal{U}_I} \textbf{S}^{(0)}_k$ and ${\textbf{S}}_k^{(0)}$ is the solution of the algorithm at the previous iteration. By doing this, as will be apparent later, the overall surrogate problem can be formulated as an SDP optimization problem.

\begin{proposition}\label{prop_sf}
A valid surrogate function, $\hat{s}_i(\bar{\textbf{S}}, \bar{\textbf{S}}^{(0)})$, for the function $s_i(\bar{\textbf{S}}) = \log \det \left(\textbf{I} + \textbf{H}_n\bar{\textbf{S}}\textbf{H}_n^H\right)$ that satisfies conditions $(\text{A}1)-(\text{A}4)$ is
\begin{equation}
\hat{s}_i(\bar{\textbf{S}}, \bar{\textbf{S}}^{(0)}) \triangleq \Tr\left(\textbf{J}_i\bar{\textbf{S}}\right) + \Tr\left(\bar{\textbf{S}}^H\textbf{M}_i \bar{\textbf{S}}\right)+ \kappa_1,\quad \forall  \bar{\textbf{S}},\, \bar{\textbf{S}}^{(0)} \in\mathcal{S}^{n_T}_+,\label{sur_apb}
\end{equation}
with matrices $\textbf{J}_i = \textbf{G}_i-\bar{\textbf{S}}^{(0),H}\textbf{M}_{i}-\textbf{M}_{i}\bar{\textbf{S}}^{(0)}$, $\textbf{G}_i = \textbf{H}_i^H\left(\textbf{I} + \textbf{H}_i\bar{\textbf{S}}^{(0)}\textbf{H}_i^H\right)^{-1}\textbf{H}_i$ and $\textbf{M}_i = -\gamma_i \textbf{I}$, being $\gamma_i \ge \frac{1}{2} \lambda_{\max}^2(\textbf{H}^H_i\textbf{H}_i)$, $\kappa_1$ contains some terms that do not depend on $\textbf{S}$, and $\mathcal{S}^{n_T}_+$ denotes the set of positive semidefinite matrices.
\end{proposition}
\begin{proof}
See Appendix \ref{app2}.
\end{proof}

Let us now reformulate the optimization problem in \eqref{op:wet1} with the surrogate function $\hat{s}_i(\bar{\textbf{S}}, \bar{\textbf{S}}^{(0)}) - \hat{g}_i(\boldsymbol\Omega_i(\textbf{S}_{-i}),\boldsymbol\Omega_i^{(0)})$:
\begin{eqnarray}
\Tr\left(\textbf{E}_i\bar{\textbf{S}}\right) + \Tr\left(\bar{\textbf{S}}^H\textbf{M}_i \bar{\textbf{S}}\right) +  \Tr\left(\textbf{R}_i\textbf{S}_i\right) + \kappa_2,\label{surg_a}
\end{eqnarray}
where $\textbf{R}_i = \textbf{H}_i^H\left(\boldsymbol\Omega_i^{(0)}\right)^{-1}\textbf{H}_i \in\mathbb{C}^{n_T\times n_T}$, $\textbf{E}_i = \textbf{J}_i - \textbf{R}_i$, and $\kappa_2$ contains some terms that do not depend on $\textbf{S}$. Thus, problem \eqref{op:wet1} can be reformulated as
\begin{alignat}{2}
\mathop{\text{max}}_{\{\textbf{S}_i\}}& \quad \sum_{i\in\mathcal{U}_I} \omega_i \Bigg(\Tr\left(\textbf{E}_i\bar{\textbf{S}}\right) + \Tr\left(\bar{\textbf{S}}^H\textbf{M}_i \bar{\textbf{S}}\right) +  \Tr\left(\textbf{R}_i\textbf{S}_i\right)\Bigg) \nonumber\\
&\quad- \rho \left\|\textbf{S}_i - \textbf{S}^{(0)}_i\right\|_F^2 \label{op:wet1bb} \\
\textrm{s. t.} 
&  \quad \textbf{S} \in\mathcal{S}_1 \nonumber,
\end{alignat}
where we have added a proximal quadratic term to the surrogate function in which $\rho$ is any non-negative constant that can be tuned by the algorithm. This term provides more flexibility in the algorithm design stage and may help to speed up the convergence. By performing some mathematical manipulations, we are able to obtain the following result:

\begin{proposition}\label{prop_p1}
The optimization problem presented in \eqref{op:wet1} can be solved based on MM method by solving recursively the following SDP problem:
\begin{alignat}{2}
\mathop{\text{min}}_{\{\textbf{S}_i\}, \,{\textbf{s}},\,t}& \quad  t \label{op:wetN3_b} \\
\textrm{s. t.} 
& \quad C1: \left[ \begin{array}{cc}
     {t}\textbf{I} & \tilde{\textbf{C}}^{\frac{1}{2}}{\textbf{s}} - \textbf{c} \\ \left(\tilde{\textbf{C}}^{\frac{1}{2}}{\textbf{s}} - \textbf{c}\right)^H & 1 \end{array} \right] \succeq 0 \nonumber \\
& \quad C2: \textbf{T}_i{\textbf{s}} = \emph{vec}\left(\textbf{S}_i\right),&&\quad \forall i \in\mathcal{U}_I\nonumber\\
&  \quad C3: \textbf{S} \in\mathcal{S}_1 \nonumber,
\end{alignat}
where ${\textbf{s}} = \left[\emph{vec}(\textbf{S}_1)^T \emph{vec}(\textbf{S}_2)^T \dots \emph{vec}(\textbf{S}_N)^T\right]^T \in\mathbb{C}^{n_Tn_T|\mathcal{U}_I|\times 1}$, $t$ is a dummy variable, and $\tilde{\textbf{C}}^{\frac{1}{2}}$, $\textbf{T}_i$, and $\textbf{c}$ are some constant matrices and vectors computed as shown in Appendix \ref{app3}. Vector $\textbf{c}$ depends on matrix $\bar{\textbf{S}}^{(0)}$.
\end{proposition}
\begin{proof}
See Appendix \ref{app3}.
\end{proof}

The final algorithm is presented in Alg. \ref{algQ1}.

\subsection{Approach to Solve the Sum Method Formulation in \eqref{op:wetre}}
\label{subsec_q2}
Let us start the development by reformulating problem \eqref{op:wetre}:
\begin{alignat}{2}
\mathop{\text{max}}_{\{\textbf{S}_i\}}& \quad\sum_{i\in\mathcal{U}_I} \omega_i\left(s_i(\textbf{S}) - \omega_ig_i(\boldsymbol\Omega_i(\textbf{S}_{-i}))\right) + \sum_{i\in\mathcal{U}_I}\Tr(\textbf{R}_H\textbf{S}_i)\nonumber \\
\textrm{s. t.} 
&  \quad \textbf{S}\in\mathcal{S}_2 \label{op:wetre2},
\end{alignat}
where $\textbf{R}_H = \sum_{j\in\mathcal{U}_E} \alpha_j\textbf{H}_j^H\textbf{H}_j$. The right hand side of the objective function of \eqref{op:wetre2} is convex (in fact it is linear) whereas the left hand side is not convex. Let us apply the same steps that we applied before but with a slight modification. Previously in \eqref{taylor}, we found that $g_i(\boldsymbol\Omega_i(\textbf{S}_{-i}))$ could be approximated by $\hat{g}_i(\boldsymbol\Omega_i(\textbf{S}_{-i}),\boldsymbol\Omega_i^{(0)}) = \Tr\left(\left(\boldsymbol\Omega_i^{(0)}\right)^{-1}\boldsymbol\Omega_i(\textbf{S}_{-i})\right)$ (omitting the constant term). Now, as the objective function is different than the one from problem \eqref{op:wet1}, the goal is to find a surrogate function for the function $s_i(\textbf{S})$ that allows us to find efficiently a solution for the surrogate problem. 
\begin{algorithm}[t]
\centering
\begin{algorithmic}[1]
\caption{Algorithm for Solving Problem \eqref{op:wet1}}
\label{algQ1}
\vspace{2 mm}
\STATE Initialize $\textbf{S}^{(0)} \in \mathcal{S}_1$. Set $k=0$\\
\STATE Repeat 
\STATE \quad Compute $\textbf{c}$ with $\textbf{S}^{(k)}$, given in \eqref{eq_c}
\STATE \quad Generate the $(k+1)$-th tuple $(\textbf{S}^\star_i)_{\forall i \in\mathcal{U}_I}$ by solving \\
 \quad the SDP in \eqref{op:wetN3_b}\\
\STATE \quad Set $\textbf{S}^{(k+1)}_i = \textbf{S}^\star_i, \,\forall i \in\mathcal{U}_I $, and set $k = k +1$\\
\STATE  Until convergence is reached\\
\end{algorithmic}
\end{algorithm}

\begin{proposition}\label{prop_sf2}
A valid surrogate function, $\hat{s}_i(\textbf{S},\textbf{S}^{(0)})$, for the function $s_i(\textbf{S})$ that satisfies conditions $(\text{A}1)-(\text{A}4)$ is 
\begin{eqnarray}
\hat{s}_i(\textbf{S},\textbf{S}^{(0)}) &\triangleq&\sum_{\ell\in\mathcal{U}_I}\Tr\left(\textbf{J}_{i}\textbf{S}_\ell\right) + \sum_{\ell\in\mathcal{U}_I}\Tr\left(\textbf{S}_\ell^H\textbf{M}_{ i} \textbf{S}_\ell\right) + \kappa_3,\nonumber\\
&&\quad \forall  \textbf{S}_\ell,\, \textbf{S}_\ell^{(0)} \in\mathcal{S}^{n_T}_+,\label{sur_ap2}
\end{eqnarray}
with matrices $\textbf{J}_i = \textbf{G}_i-\textbf{S}^{(0),H}_\ell\textbf{M}_{ i}-\textbf{M}_{ i}  \textbf{S}_\ell^{(0)}$, $\textbf{G}_{i} = \textbf{H}^H_i\left(\textbf{I} + \textbf{H}_i\sum_{k\in\mathcal{U}_I}\textbf{S}_k^{(0)}\textbf{H}_i^H\right)^{-1}\textbf{H}_i$, and $\textbf{M}_i = -\xi_i \textbf{I}$, being $\xi_i \ge \frac{1}{2} |\mathcal{U}_I|^2\lambda_{\max}^2(\textbf{H}^H_i\textbf{H}_i)$, and $\kappa_3$ contains the constant terms that do not depend on $\textbf{S}$.
\end{proposition}
\begin{proof}
See Appendix \ref{app4}.
\end{proof}

\begin{remark}
Note that the two surrogate functions \eqref{sur_apb} and \eqref{sur_ap2} have the same form but with a difference in the quadratic term. Notice that surrogate function \eqref{sur_ap2} is tighter than \eqref{sur_apb} and with cross-products. As will be shown later, this will allow us to decouple the optimization problem for each information user $i$ and, thus, solve all problems in parallel. On the other hand, thanks to the fact that surrogate function \eqref{sur_apb} is looser than \eqref{sur_ap2}, a faster convergence can be obtained than if surrogate \eqref{sur_ap2} were to be applied in problem \eqref{op:wet1}.
\end{remark}
Let us now reformulate problem \eqref{op:wetre2} with the lower bound that we just found (omitting the constant terms):
\begin{alignat}{2}
\mathop{\text{max}}_{\{\textbf{S}_i\}}& \quad \sum_{i\in\mathcal{U}_I}\Tr\left(\check{\textbf{J}}_i\textbf{S}_i\right) + \sum_{i\in\mathcal{U}_I}\Tr\left(\textbf{S}_i^H\check{\textbf{M}} \textbf{S}_i\right) \nonumber\\
&\quad- \sum_{i\in\mathcal{U}_I}\Tr\left(\textbf{R}_i\sum_{\substack{k\in\mathcal{U}_I\\k\neq i}}\textbf{S}_k\right) + \sum_{i\in\mathcal{U}_I}\Tr(\textbf{R}_H\textbf{S}_i)\label{op:wetre3} \\
\textrm{s. t.} 
&  \quad  \textbf{S}\in\mathcal{S}_2, \nonumber
\end{alignat}
where $\check{\textbf{J}}_i = \check{\textbf{G}} - \textbf{S}^{(0),H}_i\check{\textbf{M}}- \check{\textbf{M}}  \textbf{S}_i^{(0)}$, with $\check{\textbf{M}} = \sum_{k\in\mathcal{U}_I}\omega_k \textbf{M}_k$ and $\check{\textbf{G}} = \sum_{k\in\mathcal{U}_I} \omega_k\textbf{G}_k$. Note that we have arranged the indices to make the notation easier to follow and consistent with the original notation. We can further simplify the objective function by grouping terms considering that matrix $\check{\textbf{M}}$ is diagonal, i.e., $\check{\textbf{M}} = -\beta\textbf{I}$, being $\beta=\frac{1}{2} |\mathcal{U}_I|^2\sum_{k\in\mathcal{U}_I}\omega_k\lambda_{\max}^2(\textbf{H}^H_k\textbf{H}_k)$:
\begin{alignat}{2}
\mathop{\text{min}}_{\{\textbf{S}_i\}}& \quad \beta\sum_{i\in\mathcal{U}_I}\Tr\left(\textbf{S}_i^H\textbf{S}_i\right) - \sum_{i\in\mathcal{U}_I}\Tr\left(\textbf{F}_i\textbf{S}_i\right) \label{op:wetre4} \\
\textrm{s. t.} 
&  \quad  \textbf{S}\in\mathcal{S}_2, \nonumber
\end{alignat}
where 
\begin{equation}
\textbf{F}_i = \check{\textbf{J}}_i -  \sum_{\substack{k\in\mathcal{U}_I\\k\neq i}}\textbf{R}_k + \textbf{R}_H.\label{matrix_f} 
\end{equation}
Note that we have changed the sign of the objective and reformulated the problem as a minimization one. The idea is to find a closed-form expression for the optimum covariance matrices $\{\textbf{S}_i\}$. If we dualize constraint $C1$ and form a partial Lagrangian, we obtain the following optimization problem:
\begin{alignat}{2}
\mathop{\text{min}}_{\{{\textbf{S}}_i\}}& \quad \beta\sum_{i\in\mathcal{U}_I}\Tr\left({\textbf{S}}_i^H{\textbf{S}}_i\right) - \sum_{i\in\mathcal{U}_I}\Tr\left(\textbf{W}_i(\mu){\textbf{S}}_i\right) \label{op:wetre6} \\
\textrm{s. t.} 
&  \quad  {\textbf{S}}_i \succeq 0, \quad \forall i \in \mathcal{U}_I, \nonumber
\end{alignat}
where $\textbf{W}_i(\mu) = \textbf{F}_i - \mu\textbf{I}$, for $\mu\ge 0$ the Lagrange multiplier associated with constraint $C1$ of problem \eqref{op:wetre2}. The previous problem is clearly separable for each user $i$. Thus, for each information user, problem \eqref{op:wetre6} is equivalent to solving the following projection problem:
\begin{alignat}{2}
\mathop{\text{min}}_{{\textbf{S}}_i}& \quad \left\|\sqrt{\beta}{\textbf{S}}_i - \check{\textbf{W}}_i(\mu)\right\|_F \label{op:wetre7} \\
\textrm{s. t.} 
&  \quad {\textbf{S}}_i \succeq 0,  \nonumber
\end{alignat}
where $\check{\textbf{W}}_i(\mu) = \frac{1}{2\sqrt{\beta}}\textbf{W}_i(\mu) = \frac{1}{2\sqrt{\beta}}(\textbf{F}_i - \mu\textbf{I})$. The previous result is very nice as the solution of \eqref{op:wetre7} is simple and elegant, thanks to the fact that problem \eqref{op:wetre7} is a projection onto the semidefinite cone and has a closed-form solution \cite{henrion:11}.
Let the eigenvalue decomposition (EVD) of matrix $\textbf{F}_i$ be $\textbf{F}_i = \textbf{U}_{F_i} \boldsymbol\Lambda_{F_i}\textbf{U}_{F_i}^H$. The expression of $\textbf{S}_i^\star(\mu)$ is, thus, given by
\begin{equation}
\textbf{S}_i^\star(\mu) = \frac{1}{\sqrt{\beta}} [\check{\textbf{W}}_i(\mu)]^+ = \frac{1}{2\beta}\textbf{U}_{F_i}^H[\boldsymbol\Lambda_{F_i} - \mu\textbf{I}]^+\textbf{U}_{F_i},\quad \forall i \in\mathcal{U}_I,
\end{equation}
where $\lambda_k ([\textbf{X}]^+) = \min(0,\lambda_k(\textbf{X}))$, with $\lambda_k(\textbf{X})$ the $k$-th eigenvalue of matrix $\textbf{X}$. Now it remains to compute the optimal Lagrange multiplier $\mu$. This can be found by means of the simple bisection method fulfilling $\sum_{i\in\mathcal{U}_I}\Tr\left([\boldsymbol\Lambda_{F_i} - \mu\textbf{I}]^+\right) = 2\beta P_T$. It turns out that, at each inner iteration, we need to compute a single EVD per information user, that is, the EVD of $\textbf{F}_i$, and a few iterations to find the optimal multiplier $\mu$. Note that the surrogate problem can be solved straightforwardly with the previous steps. The final algorithm is presented in Alg. \ref{algQ2}.
\begin{algorithm}[t]
\centering
\begin{algorithmic}[1]
\caption{Algorithm for Solving Problem \eqref{op:wetre}}
\label{algQ2}
\vspace{2 mm}
\STATE Initialize $\textbf{S}^{(0)} \in \mathcal{S}_2$. Set $k=0$\\
\STATE Repeat
\STATE \quad Compute $\textbf{F}_i$ with matrix $\textbf{S}_i^{(k)}$, $\forall i \in \mathcal{U}_I$, given in \eqref{matrix_f}
\STATE \quad Compute EVD of $\textbf{F}_i = \textbf{U}_{F_i} \boldsymbol\Lambda_{F_i}\textbf{U}_{F_i}^H$, \quad $\forall i \in \mathcal{U}_I$
\STATE \quad Compute $\mu^\star$ such that \\
\quad $\sum_{i\in\mathcal{U}_I}\Tr\left([\boldsymbol\Lambda_{F_i} - \mu^\star\textbf{I}]^+\right) = 2\beta P_T$
\STATE \quad Compute $\textbf{S}^\star_i(\mu^\star)=\frac{1}{2\beta}[\textbf{F}_i - \mu^\star\textbf{I}]^+,\quad \forall i\in\mathcal{U}_I$\\
\STATE \quad Set $\textbf{S}^{(k+1)}_i = \textbf{S}^\star_i(\mu^\star), \forall i \in \mathcal{U}_I$, and set $k = k +1$\\
\STATE  Until convergence is reached\\
\end{algorithmic}
\end{algorithm}

\subsection{Approaches Used as Benchmarks for Performance Comparison}
As the problem introduced in \eqref{op:mob} has not been addressed before in the literature, there are not specific benchmarks to compare our approaches with. For this reason, in this section, we propose some benchmark algorithms that will be used in the simulation section to compare the performance of the proposed MM approaches. These benchmarks are:
\begin{itemize}
\item Gradient-based algorithms based on \cite[Sec. 7]{grad_notes} applied directly to the nonconvex problems \eqref{op:wet1} and \eqref{op:wetre}. The gradients are not presented due to space limitations.
\item MM approaches considering just the linear approximation presented in \eqref{taylor}, i.e., $\hat{g}_i(\boldsymbol\Omega_i(\textbf{S}_{-i}),\boldsymbol\Omega_i^{(0)})$, applied to problems \eqref{op:wet1} and \eqref{op:wetre}. The specific optimization problems and algorithms can be found in App. \ref{app_bench}.
\end{itemize}

\section{Numerical Evaluation}
\label{sec_num_mm}
In this section, we evaluate the performance of the previous algorithms. In the first part of this section, we present some convergence and computational time results. For the simulations, we consider a system composed of 1 transmitter with 6 antennas, and 3 information users and 3 harvesting users with 2 antennas each. In the second part of the section, we show the performance of the proposed methods compared to the classical BD approach. In this case, for ease of presenting the information, we assume a system composed of 1 transmitter with 4 antennas, and 2 information users and 2 harvesting users with 2 antennas each. The simulation parameters common to both scenarios are the following. The maximum radiated power is $P_{T} = 1$ W. The channel matrices are generated randomly with i.i.d. entries distributed according to $\mathcal{CN} (0,1)$. The weights $\omega_i$ are set to 1.
\subsection{Convergence Evaluation}
In this subsection, we evaluate the convergence behavior and the computational time of the methods presented in Sections \ref{subsec_q1} and \ref{subsec_q2} and the benchmark approach presented in App. \ref{app_bench}. The benchmark method for problem \eqref{op:wetre} presented in App. \ref{app_bench} will not be evaluated as it is clearly worse\footnote{However, it was included in the paper for the sake of completeness} than the one presented in Section \ref{subsec_q2}. In the figures, the legend is interpreted as follows: `MM-L for \eqref{op:wet1}' refers to the method developed in App. \ref{app_bench} for problem \eqref{op:wet1}, `MM-Q for \eqref{op:wet1}' refers to the method in Section \ref{subsec_q1}, and `MM-Q for \eqref{op:wetre}' refers to the method in Section \ref{subsec_q2}. In order to compare all methods, we set the values of $\alpha_j$ and the values of $Q_j$ so that the same system sum rate is achieved. These values are: $\boldsymbol\alpha = [1, 5, 10]$, and $\textbf{Q} = [3.8, 7.2, 6.4]$ power units. Software package CVX is used to solve problem \eqref{op:wet3} \cite{cvx}, and SeDuMi solver is used to solve problem \eqref{op:wetN3_b} \cite{sedumi}.

Figure \ref{fig1} presents the sum rate convergence as a function of iterations. The three approaches converge to the same sum rate value but require a different number of iterations. In fact, the required number of iterations depends on how well the surrogate function approximates the original function. Note that the surrogate function used in the `MM-L for \eqref{op:wet1}' approach is the one that best approximates the objective function and, thus, fewer iterations are needed. 

\begin{figure}
 \centering
\includegraphics[width=0.5\textwidth]{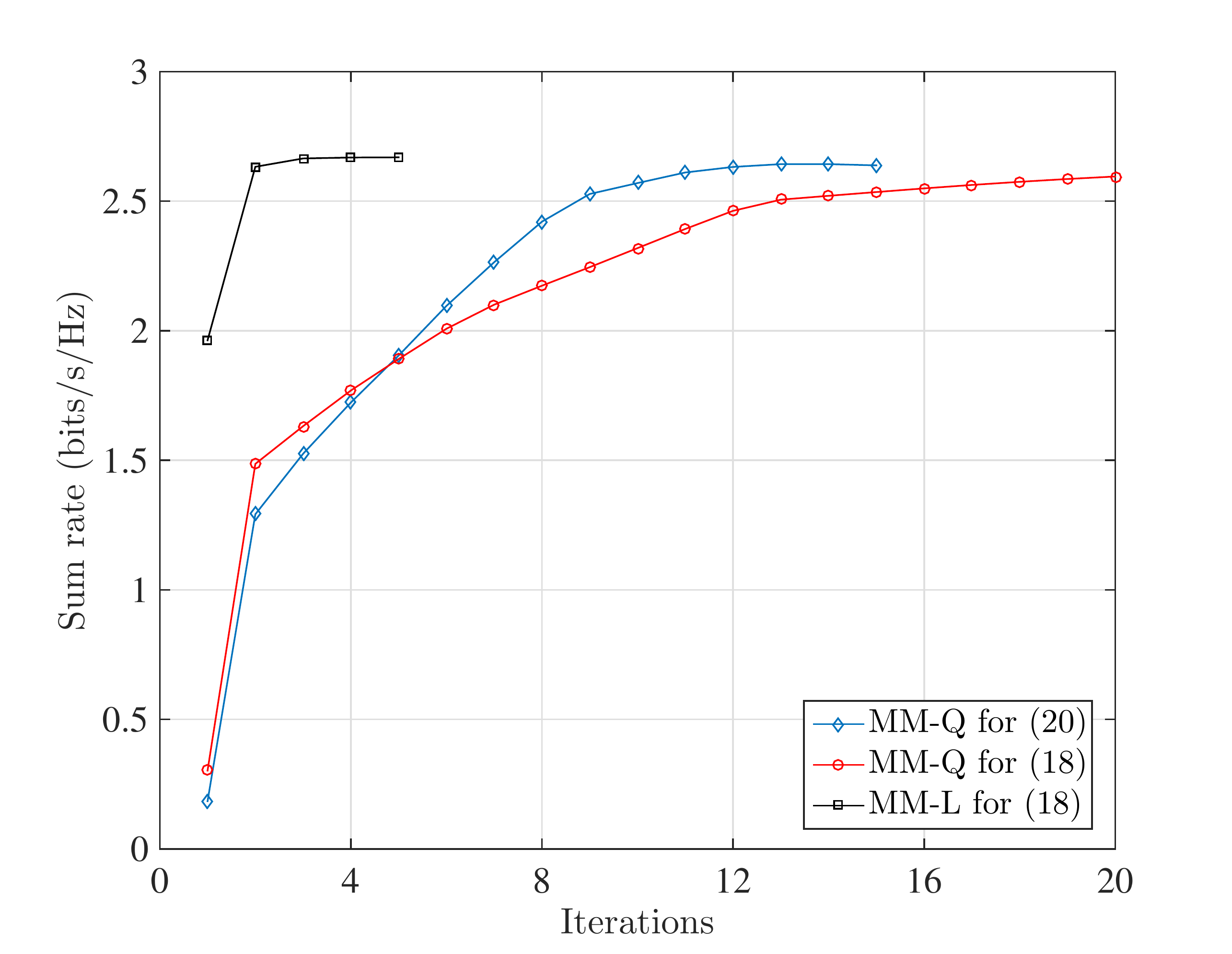}  
\caption{Convergence of the system sum rate vs number of iterations for three different approaches.}
\label{fig1}
\vspace{-4mm}
\end{figure}

\begin{figure}
 \centering
\includegraphics[width=0.5\textwidth]{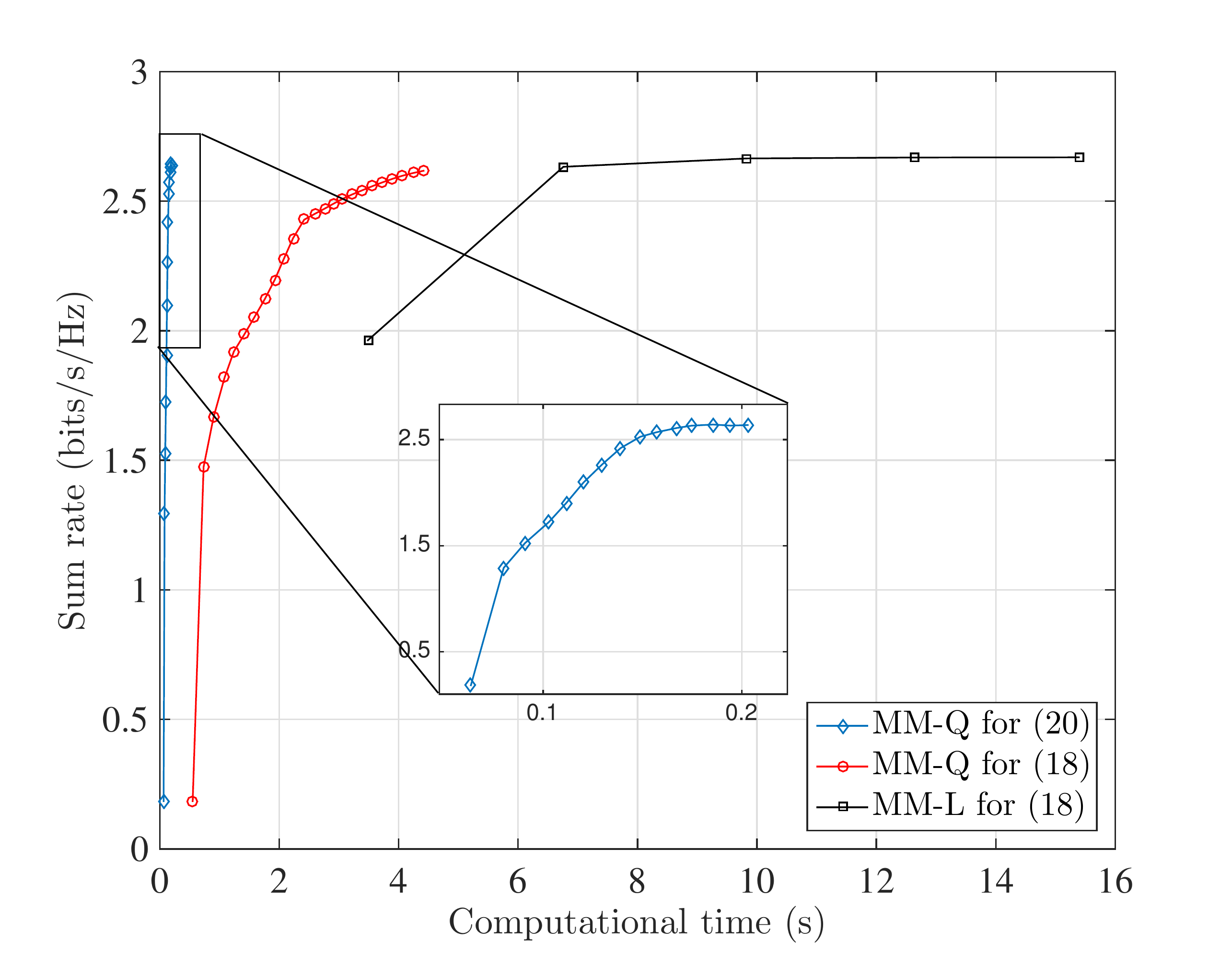}  
\caption{Convergence of the system sum rate vs computational time for three different approaches.}
\label{fig1b}
\vspace{-4mm}
\end{figure}

Figure \ref{fig1b} shows the computational time required by the three previous methods. We see that the `MM-Q for \eqref{op:wetre}' method converges much faster than the other two approaches, as expected. The `MM-Q for \eqref{op:wet1}' approach requires more iterations than the `MM-L for \eqref{op:wet1}' approach but each iteration is solved faster since a specific algorithm can be employed to solve the convex optimization problem. Hence, the `MM-Q for \eqref{op:wet1}' algorithm is the best option.

For the sake of comparison and completeness, we also show in Figures \ref{fig1c} and \ref{fig1d} the convergence and the computational time of a gradient-like benchmark approach. The plot legend reads as follows: `GRAD for \eqref{op:wet1}' and `GRAD for \eqref{op:wetre}' refers to a gradient approach applied to problems \eqref{op:wet1} and \eqref{op:wetre}, respectively. `all ones' and `identity' mean that covariance matrices are initialized using an all ones matrix and the identity matrix, respectively. Results show that the proposed MM approaches are one to two orders of magnitude faster than the gradient-based methods.  
\begin{figure}
 \centering
\includegraphics[width=0.5\textwidth]{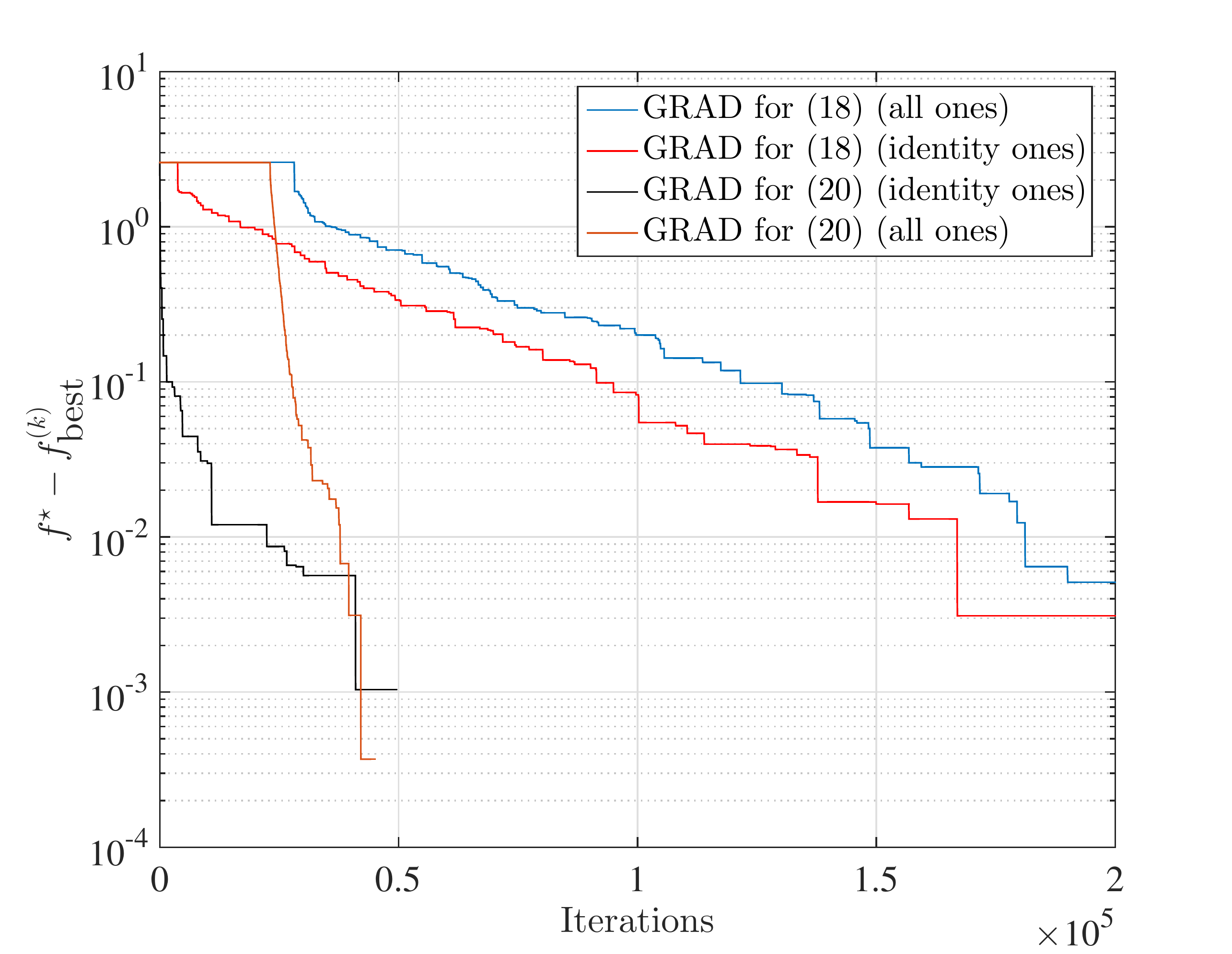}  
\caption{Convergence of the system sum rate vs iterations for a gradient approach for constrained optimization.}
\label{fig1c}
\vspace{-1mm}
\end{figure}

\begin{figure}
 \centering
\includegraphics[width=0.5\textwidth]{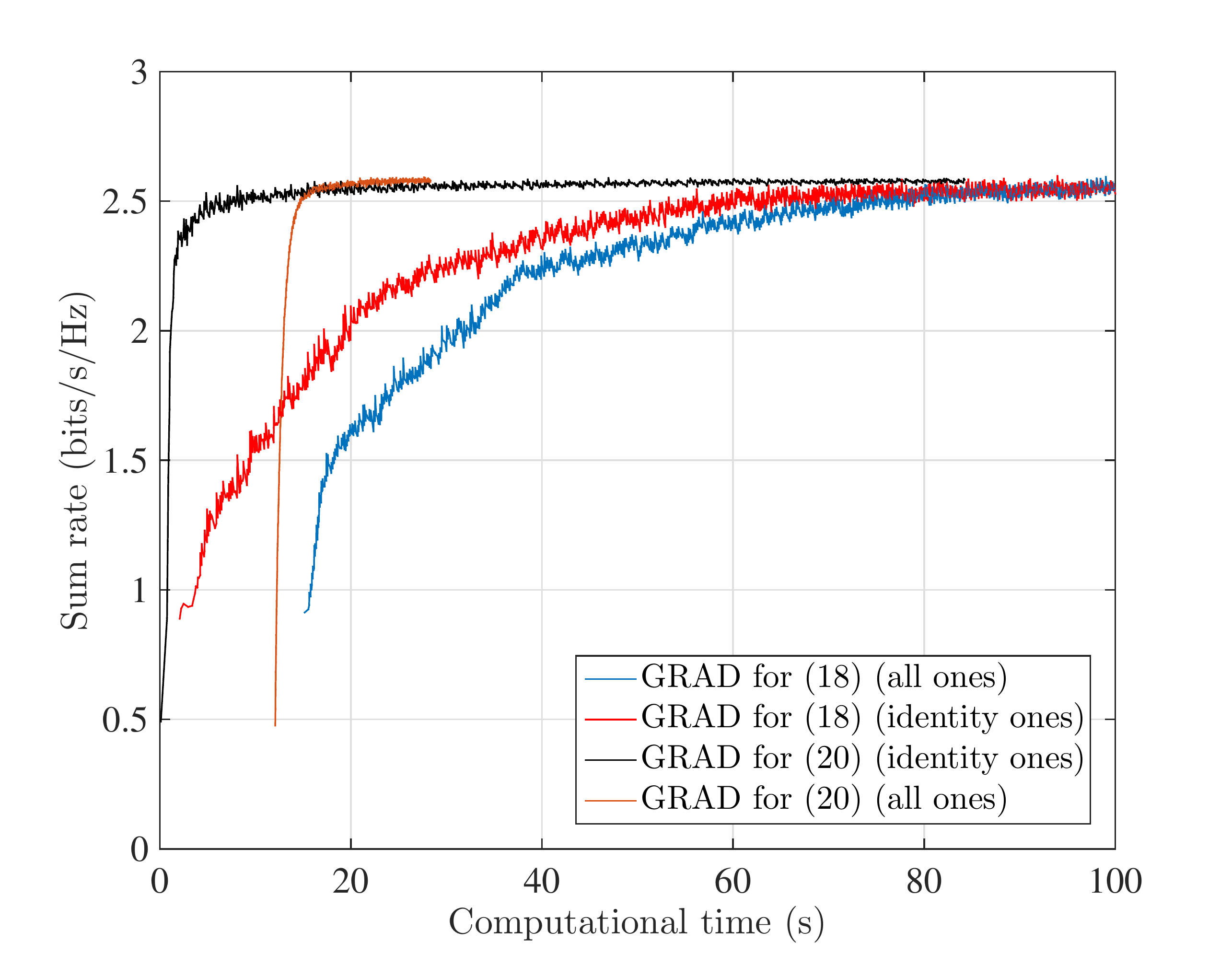}  
\caption{Convergence of the system sum rate vs computational time for a gradient approach for constrained optimization.}
\label{fig1d}
\vspace{-5mm}
\end{figure}
\begin{figure}
 \centering
\includegraphics[width=0.5\textwidth]{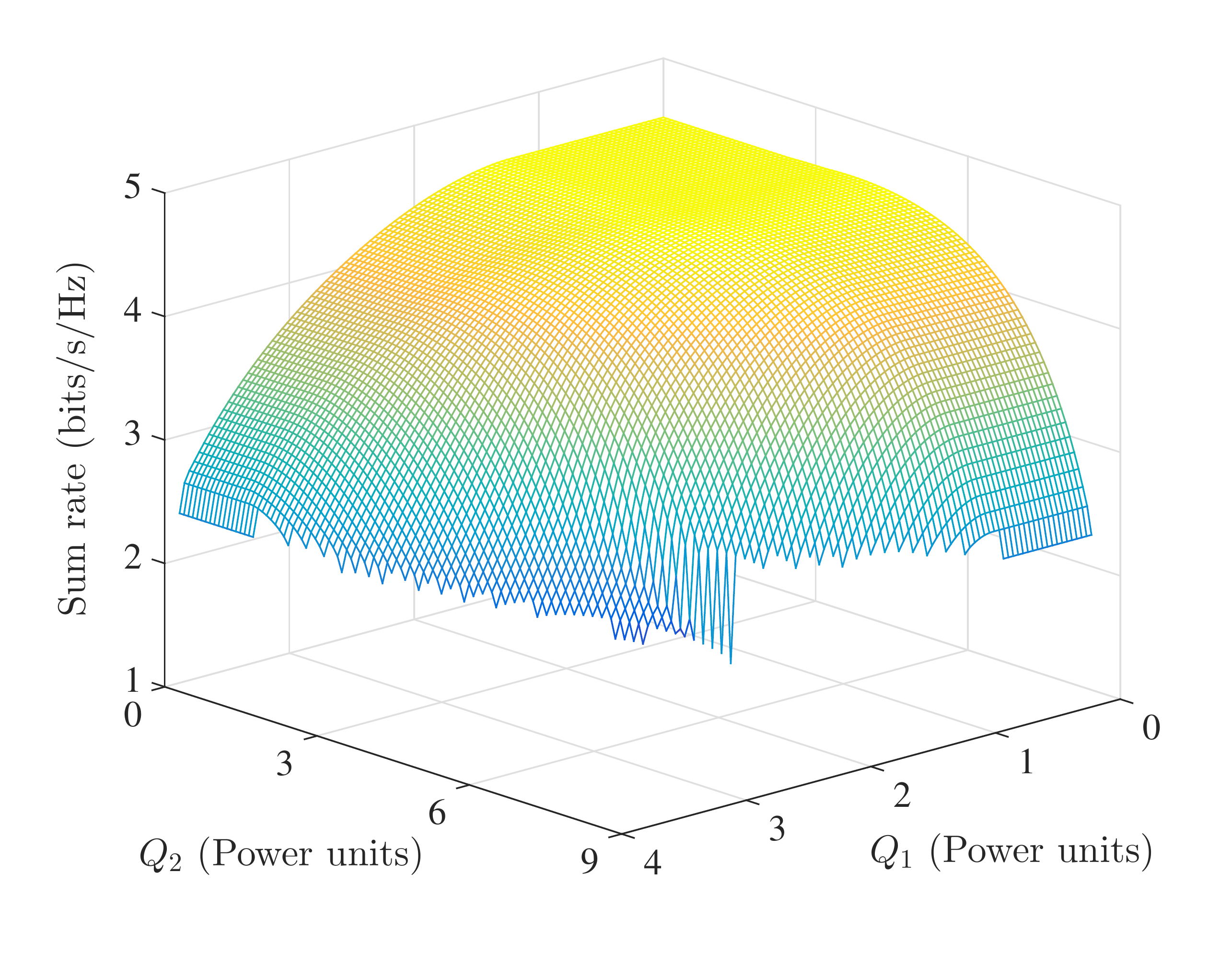}  
\caption{Rate-power surface for the MM method.}
\label{fig2}
\vspace{-6mm}
\end{figure}

\subsection{Performance Evaluation}
In this section, we evaluate the performance of the MM approach as compared to the classical BD strategy considered in the literature (see, for example, \cite{rub:13d}, \cite{zhang_m:10a}). In order to show how harvesting users at different distances affect the performance,
we have generated channel matrices with different norms. We would like to emphasize that, as the noise and channels are normalized, we will refer to the powers harvested by the receivers in terms of power units instead of Watts.

Figures \ref{fig2} and \ref{fig3} show the rate-power surface, that is, the multidimensional trade-off between the system sum rate and the powers to be collected by harvesting users (see \cite{rub:13d} for a formal definition of the rate-power surface). As we see, the MM approach outperforms the BD strategy in both terms, sum rate and harvested power. The maximum system sum rate obtained with the MM approach when $Q_1$ and $Q_2$ are set to 0 is $4.5$ bit/s/Hz, whereas the sum rate obtained with the BD approach is $2.75$ bit/s/Hz. The rate-power surfaces are generated by varying the values of $\{Q_j\}$ in problem \eqref{op:wet1} or, equivalently, by varying the values of $\{\alpha_j\}$ in problem \eqref{op:wetre}. A way to reduce the computational complexity associated with the generation of the rate-power surface is to use as an initialization point the solution that was obtained for the previous values of $\{Q_j\}$ or $\{\alpha_j\}$ to generate the new value of the curve \cite{boyd_admm}. Note, however, that the whole rate-power surface need not be generated for each transmission as it is just the representation of the existing rate-power tradeoff.

\begin{figure}
 \centering
\includegraphics[width=0.5\textwidth]{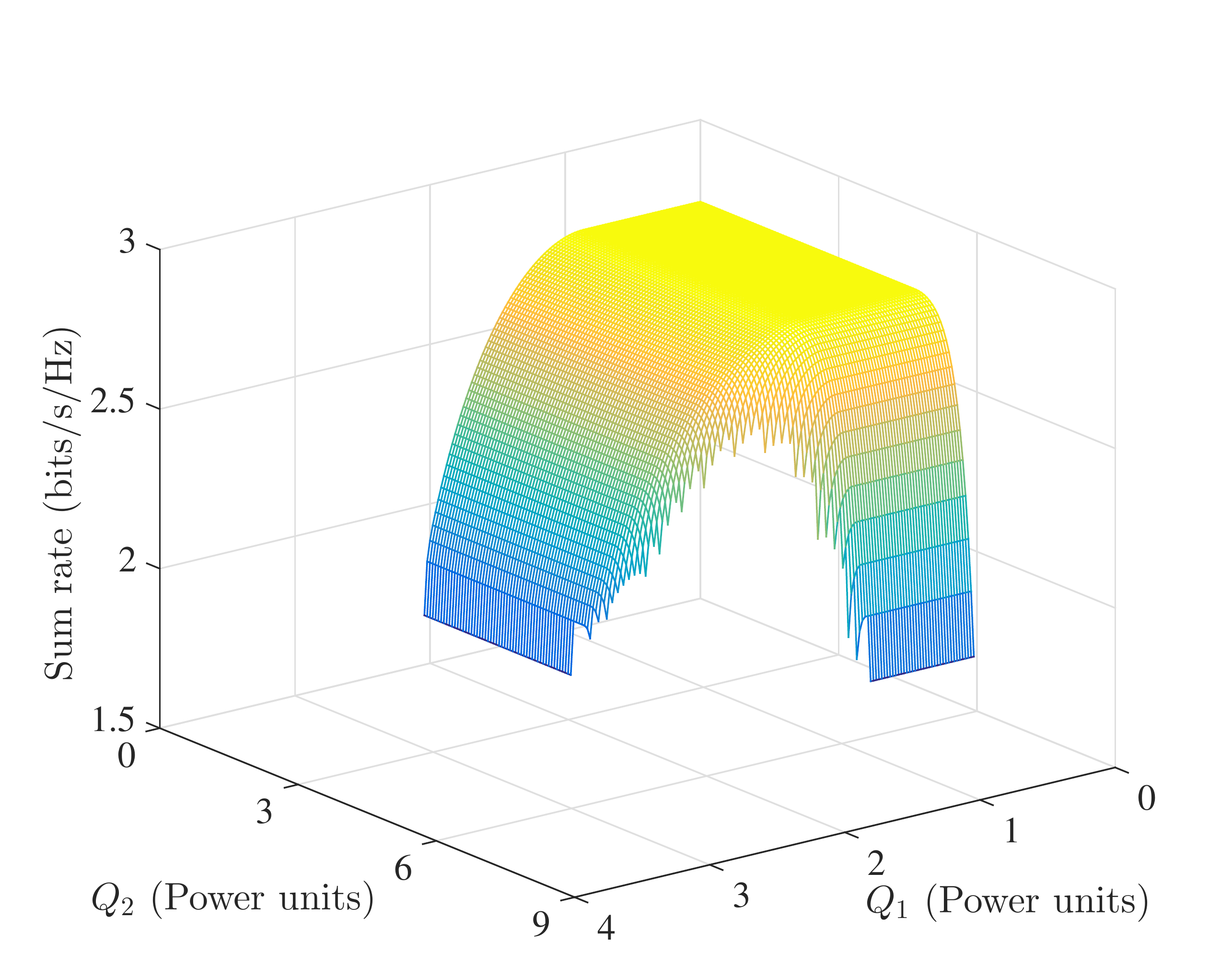}  
\caption{Rate-power surface for the BD method.}
\label{fig3}
\vspace{-4mm}
\end{figure}

\begin{figure}[t]
 \centering
\includegraphics[width=0.5\textwidth]{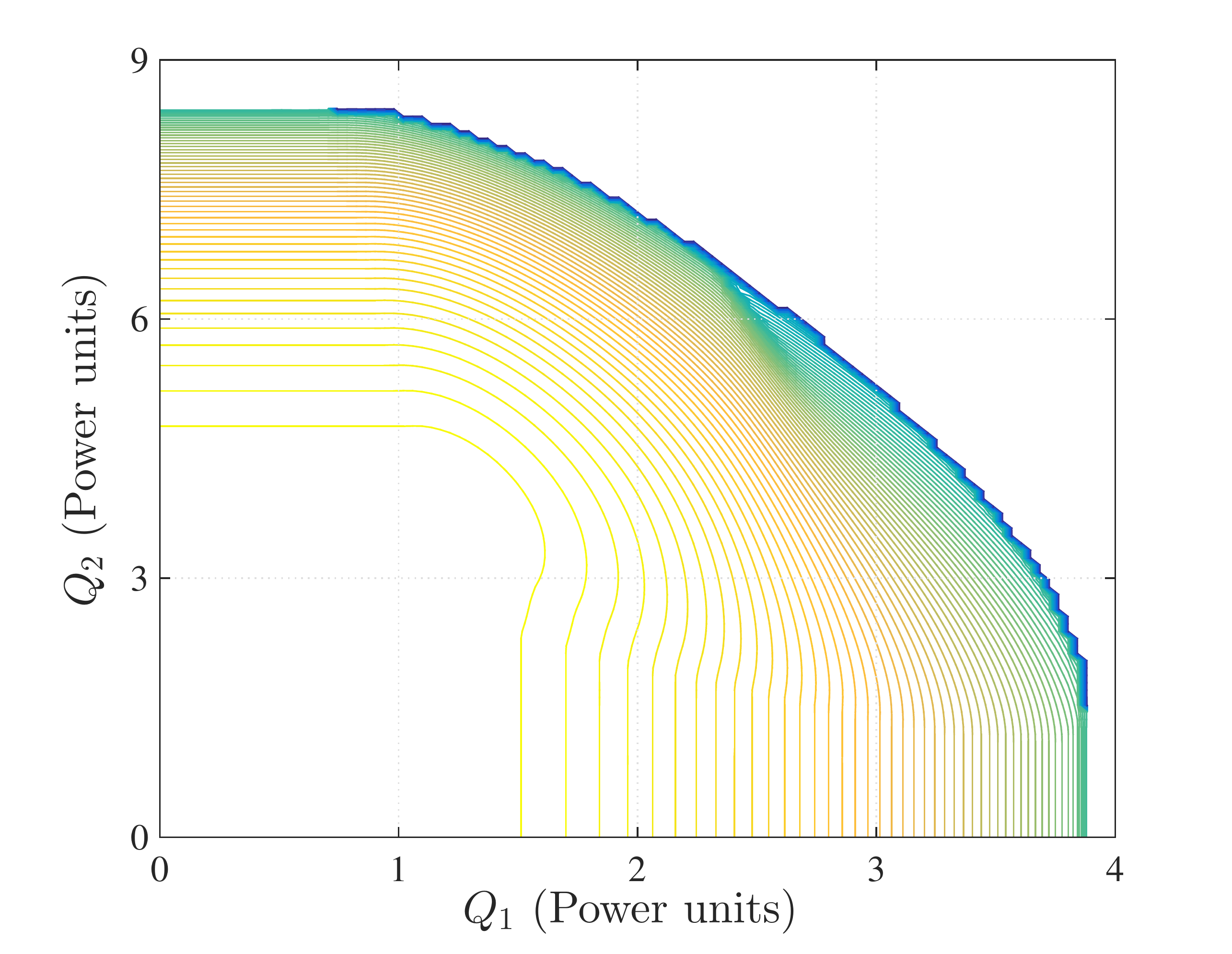}  
\caption{Contour of rate-power surface for the MM method.}
\label{fig4}
\vspace{-5mm}
\end{figure}

In order to clearly see the benefits in terms of collected power, Figures \ref{fig4} and \ref{fig5} show the contour plots of the previous 3D plots. We observe that users in the MM approach collect roughly $50\%$ more power than the power collected by users when applying the BD strategy.

\begin{figure}[t]
 \centering
\includegraphics[width=0.5\textwidth]{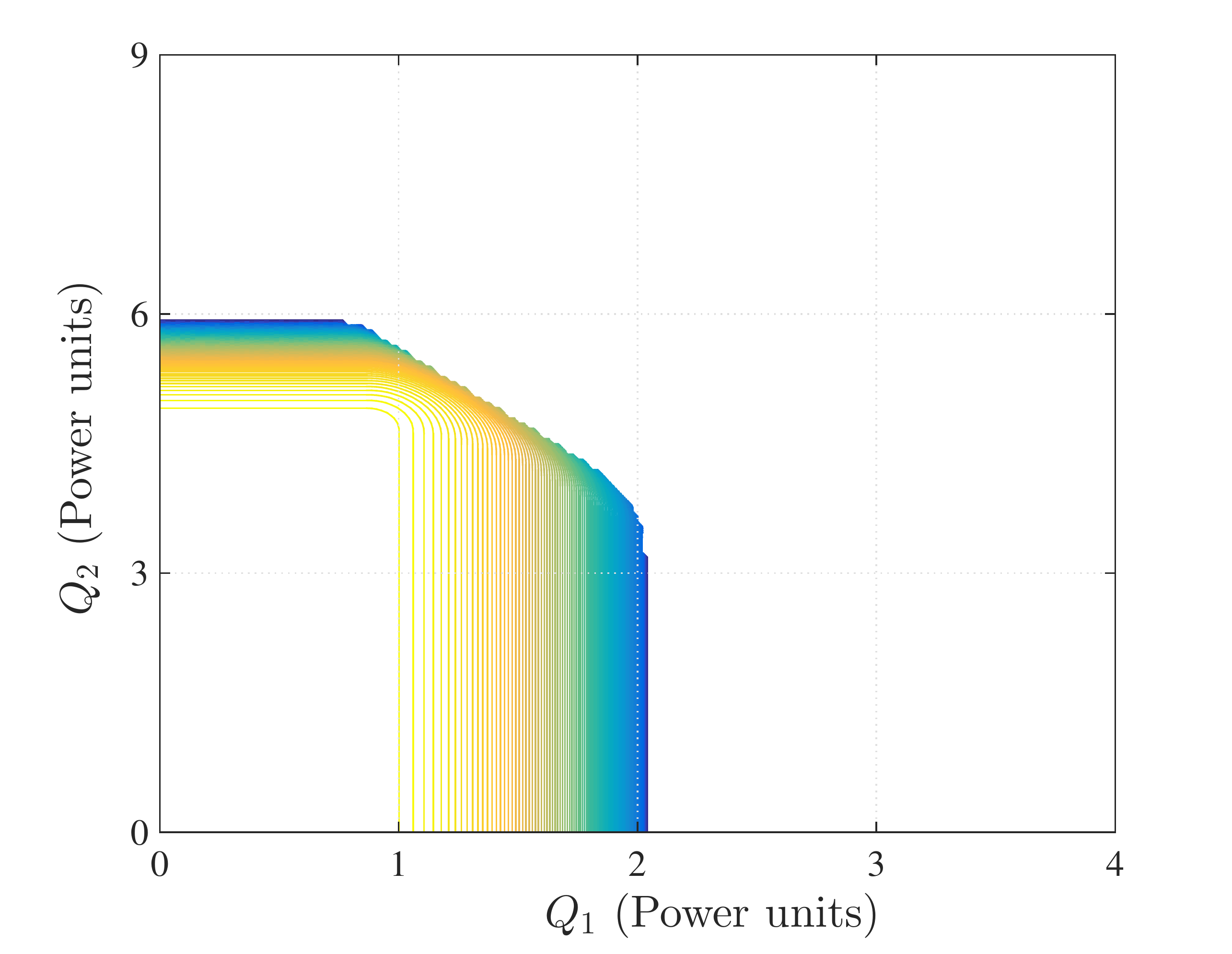}  
\caption{Contour of rate-power surface for the BD method.}
\label{fig5}
\vspace{-2mm}
\end{figure}

Finally, Figure \ref{fig7} presents the rate-region of the MM approach for different values of $\{Q_j\}$. The same value of $Q_j$ is set to the two harvesting users. In this case, we vary the values of $\omega_i$ to achieve the whole contour of the rate regions. We observe that, the larger the harvesting constraints, the smaller the rate-region, as expected. However, the relation between the harvesting constraints and the rate-region is not linear. As the harvesting constraints increase, a small change in the $\{Q_j\}$ produces a large reduction of the rate-region. This is because the 3D rate-power surfaces presented before are not planes.
\begin{figure}
 \centering
\includegraphics[width=0.44\textwidth]{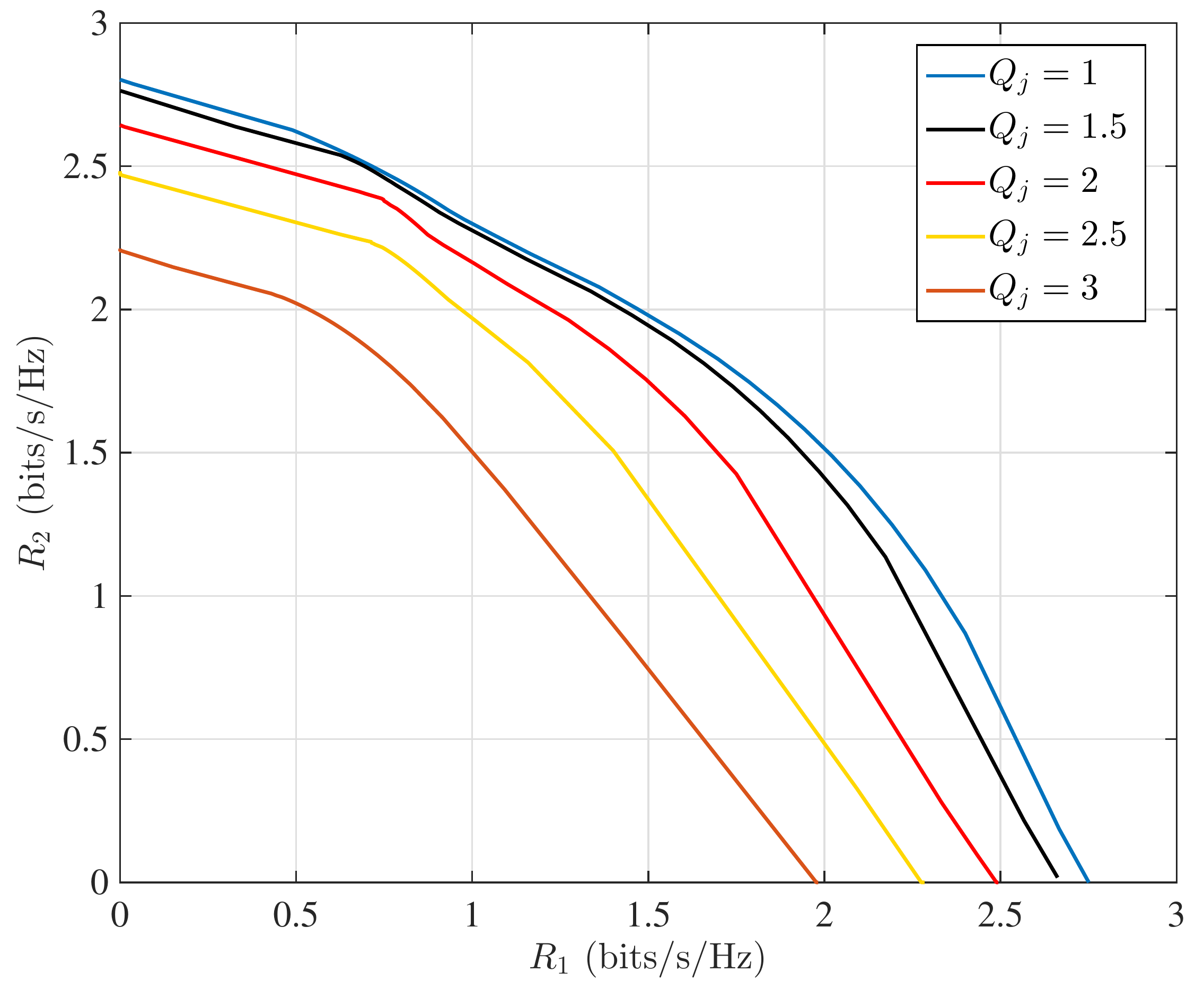}  
\caption{Rate region for different values of $Q_j$ (in power units).}
\label{fig7}
\vspace{-6mm}
\end{figure}

\section{Conclusions}
We have presented a method to solve the difficult nonconvex problem that arises in multiuser multi-stream broadcast MIMO SWIPT networks. We formulated the general SWIPT problem as a multi-objective optimization problem, in which rates and harvested powers were to be optimized simultaneously. Then, we proposed two different formulations to obtain solutions of the general multi-objective optimization problem depending on the desired level of control of the power to be harvested. In the first approach, the transmitter was able to control the specific amount of power to be harvested by each user whereas in the second approach only the proportions of power to be harvested among the different users could be controlled. Both (nonconvex) formulations were solved based on the MM approach. We derived a convex approximation for two nonconvex objectives and developed two different algorithms. Simulation results showed that the proposed methods outperform the classical BD in terms of both system sum rate and power collected by users by a factor of approximately $50\%$. Moreover, the computational time needed to achieve convergence was shown to be really low for the approach in which the transmitter could only control the proportion of powers to be harvested (around two orders of magnitude lower than a gradient-like approach).
\label{sec_con_mm}

\appendices
\section{Benchmark Formulations and Algorithms}
\label{app_bench}
In this appendix, we are going to describe the benchmarks based on the works in \cite{Hong:16}, \cite{scutari:14}, and \cite{You:14}. We start with the benchmark for problem \eqref{op:wet1}.

Note that the upper bound $\hat{g}_i(\boldsymbol\Omega_i(\textbf{S}_{-i}),\boldsymbol\Omega_i^{(0)})$ can be used to build a lower bound of $f_0(\bar{\textbf{S}})$ that fulfills the four conditions $(\text{A}1) - (\text{A}4)$ presented before in Section \ref{sec_mm}.

By applying a successive approximation of $f_0(\cdot)$ through the application of the previous surrogate function, i.e., $\hat{f}_0(\textbf{S},\textbf{S}^{(k)}) =  \sum_{i\in\mathcal{U}_I} \omega_is_i(\textbf{S}) - \omega_i\hat{g}_i(\boldsymbol\Omega_i(\textbf{S}_{-i}),\boldsymbol\Omega_i^{(k)}) - \rho \left\|\textbf{S}_i - \textbf{S}^{(k)}_i\right\|_F^2$, where $\textbf{S}^{(k)} \triangleq (\textbf{S}_i^{(k)})_{\forall i \in\mathcal{U}_I}$, for different evaluation points, we obtain an iterative algorithm based on the MM approach that converges to a stationary point (or local optimum) of the original problem \eqref{op:wet1}. Note that we have considered a proximal-like term. Given this, the convex optimization problem to solve is
\begin{alignat}{2}
\mathop{\text{max}}_{\{\textbf{S}_i\}}& \quad \sum_{i\in\mathcal{U}_I} \omega_is_i(\textbf{S}) - \omega_i\hat{g}_i(\boldsymbol\Omega_i(\textbf{S}_{-i}),\boldsymbol\Omega_i^{(k)}) - \rho \left\|\textbf{S}_i - \textbf{S}^{(k)}_i\right\|_F^2 \nonumber\\
\textrm{s. t.} 
&  \quad \textbf{S} \in \mathcal{S}_1 \label{op:wet3}.
\end{alignat}
We must proceed iteratively until convergence is reached. The procedure is presented in Alg. \ref{alg1}.

Let us now continue with the benchmark for problem \eqref{op:wetre}. If we apply the bound from \eqref{taylor}, i.e., $\hat{g}_i(\boldsymbol\Omega_i(\textbf{S}_{-i}),\boldsymbol\Omega_i^{(0)})$, problem \eqref{op:wetre} can be solved by solving consecutively the following problem:
 \begin{alignat}{2}
 \mathop{\text{max}}_{\{\textbf{S}_i\}}& \quad\sum_{i\in\mathcal{U}_I} \omega_is_i(\textbf{S}) - \omega_i\hat{g}_i(\boldsymbol\Omega_i(\textbf{S}_{-i}),\boldsymbol\Omega_i^{(k)}) + \Tr(\textbf{R}_H\textbf{S}_i) \nonumber\\
&\quad- \rho \left\|\textbf{S}_i - \textbf{S}^{(k)}_i\right\|_F^2\label{op:wetre2p} \\
 \textrm{s. t.} 
 &  \quad \textbf{S}\in\mathcal{S}_2 \nonumber.
 \end{alignat}
As problem \eqref{op:wetre2p} is convex, the MM method can be invoked to obtain a local optimum of problem \eqref{op:wetre}, following the same procedure as we did before for problem \eqref{op:wet3}.

\section{Proof of Proposition \ref{prop_sf}}
\label{app2}
The proposed quadratic surrogate function of $s_i(\bar{\textbf{S}})$ has the following form:
\begin{eqnarray}
\hspace{-5mm}\hat{s}_i(\bar{\textbf{S}}, \bar{\textbf{S}}^{(0)}) &\triangleq& \log\det\left(\textbf{I} + \textbf{H}_i\bar{\textbf{S}}^{(0)}\textbf{H}_i^H\right) \nonumber\\
&&+ \text{Re}\Big\{\Tr\left(\textbf{G}_i\left(\bar{\textbf{S}} - \bar{\textbf{S}}^{(0)}\right)\right)\Big\} \nonumber\\
&&+ \Tr\left(\left(\bar{\textbf{S}} - \bar{\textbf{S}}^{(0)}\right)^H\textbf{M}_i \left(\bar{\textbf{S}} - \bar{\textbf{S}}^{(0)}\right)\right)\nonumber\label{sur_ap}\\
&\le&   \log\det\left(\textbf{I} + \textbf{H}_i\bar{\textbf{S}}\textbf{H}_i^H\right), \quad \forall  \bar{\textbf{S}},\, \bar{\textbf{S}}^{(0)} \in\mathcal{S}^{n_T}_+,
\end{eqnarray}
where matrices $\textbf{G}_i \in\mathbb{C}^{n_T\times n_T}$ and $\textbf{M}_i \in\mathbb{C}^{n_T\times n_T}$ need to be found such that conditions $(\text{A}1)$ through $(\text{A}4)$ are satisfied, and $\text{Re}\{x\}$ denotes the real part of $x$. Note that $(\text{A}1)$ and $(\text{A}4)$ are already satisfied. Only $(\text{A}2)$ and $(\text{A}3)$ must be ensured. 
\begin{algorithm}[t]
\centering
\begin{algorithmic}[1]
\caption{Algorithm for Solving Problem \eqref{op:wet1}}
\label{alg1}
\vspace{2 mm}
\STATE Initialize $\textbf{S}^{(0)} \in \mathcal{S}_1$. Set $k=0$\\
\STATE Repeat 
\STATE \quad Generate the $(k+1)$-th tuple $(\textbf{S}^\star_i)_{\forall i \in\mathcal{U}_I}$ by solving \eqref{op:wet3}\\
\STATE \quad Set $\textbf{S}^{(k+1)}_i = \textbf{S}^\star_i, \,\forall i \in\mathcal{U}_I$, and set $k = k +1$\\
\STATE  Until convergence is reached\\
\end{algorithmic}
\end{algorithm}

Let us start by proving condition $(\text{A}3)$. Let $\bar{\textbf{S}}^{(0)}$ and $\bar{\textbf{S}}^{(1)}$ be two positive semidefinite matrices, i.e, $\bar{\textbf{S}}^{(0)}$, $\bar{\textbf{S}}^{(1)} \in\mathcal{S}^{n_T}_+$. Then, the directional derivative of the surrogate function $\hat{s}_i(\bar{\textbf{S}}, \bar{\textbf{S}}^{(0)})$ in \eqref{sur_ap} at $\bar{\textbf{S}}^{(0)}$ with direction $\bar{\textbf{S}}^{(1)} - \bar{\textbf{S}}^{(0)}$ is given by:
\begin{equation}
\text{Re}\Big\{\Tr\left(\textbf{G}_i\left(\bar{\textbf{S}}^{(1)} - \bar{\textbf{S}}^{(0)}\right)\right)\Big\}.\label{2direc_A}
\end{equation}

Now, let us compute the directional derivative of the term $\log\det\left(\textbf{I} + \textbf{H}_i\bar{\textbf{S}}\textbf{H}_i^H\right)$:
\begin{equation}
\Tr\left(\textbf{H}^H_i\left(\textbf{I} + \textbf{H}_i\bar{\textbf{S}}^{(0)}\textbf{H}_i^H\right)^{-1}\textbf{H}_i\left(\bar{\textbf{S}}^{(1)} - \bar{\textbf{S}}^{(0)}\right)\right),\label{2direc_B}
\end{equation}
where we have used $\text{d} \log\det(\textbf{X}) = \Tr(\textbf{X}^{-1}\text{d}\textbf{X})$ \cite{magnus}. Hence, by applying condition $(\text{A}3)$, the two directional derivatives \eqref{2direc_A} and \eqref{2direc_B} must be equal, from which we are able to identify matrix $\textbf{G}_i$ as
\begin{equation}
\textbf{G}_i = \textbf{H}_i^H\left(\textbf{I} + \textbf{H}_i\bar{\textbf{S}}^{(0)}\textbf{H}_i^H\right)^{-1}\textbf{H}_i, \quad \quad \textbf{G}_i = \textbf{G}_i^H.
\end{equation}

\begin{figure*}
\begin{eqnarray}
&& \log\det\left(\textbf{I} + \textbf{H}_i\bar{\textbf{S}}^{(0)}\textbf{H}_i^H\right) + \mu \Tr\left(\textbf{G}_i\left(\bar{\textbf{S}}^{(1)} - \bar{\textbf{S}}^{(0)}\right)\right) + \mu^2 \Tr\left(\left(\bar{\textbf{S}}^{(1)} - \bar{\textbf{S}}^{(0)}\right)^H\textbf{M}_i \left(\bar{\textbf{S}}^{(1)} - \bar{\textbf{S}}^{(0)}\right)\right)\nonumber\\
&&\le   \log\det\left(\textbf{I} + \textbf{H}_i\left(\bar{\textbf{S}}^{(0)} + \mu \left(\bar{\textbf{S}}^{(1)} - \bar{\textbf{S}}^{(0)}\right)\right)\textbf{H}_i^H\right), \quad \forall  \bar{\textbf{S}}^{(1)},\,\bar{\textbf{S}}^{(0)} \in\mathcal{S}^{n_T}_+, \,\forall \mu \in [0,1].\label{lin_ap}
\end{eqnarray}
\hrulefill
\begin{eqnarray}
2\Tr\Bigg(\Big(\bar{\textbf{S}}^{(1)} &-& \bar{\textbf{S}}^{(0)}\Big)^H\textbf{M}_i \left(\bar{\textbf{S}}^{(1)} - \bar{\textbf{S}}^{(0)}\right)\Bigg) \le \frac{\partial^2}{\partial \mu^2}  \log\det\left(\textbf{I} + \textbf{H}_i\left(\bar{\textbf{S}}^{(0)} + \mu \left(\bar{\textbf{S}}^{(1)} - \bar{\textbf{S}}^{(0)}\right)\right)\textbf{H}_i^H\right)\Bigg|_{\forall  \bar{\textbf{S}}^{(1)},\bar{\textbf{S}}^{(0)} \in\mathcal{S}^{n_T}_+, \,\forall \mu \in [0,1]}.\label{tayl_b}
\end{eqnarray}
\hrulefill
\begin{eqnarray}
\frac{\partial}{\partial \mu} \log\det\Big(\textbf{I} &+& \textbf{H}_i\Big(\bar{\textbf{S}}^{(0)} + \mu \left(\bar{\textbf{S}}^{(1)} - \bar{\textbf{S}}^{(0)}\right)\Big)\textbf{H}_i^H\Big)\nonumber\\
&=& \Tr\left(\left(\textbf{I} + \textbf{H}_i\Big(\bar{\textbf{S}}^{(0)} + \mu \left(\bar{\textbf{S}}^{(1)} - \bar{\textbf{S}}^{(0)}\right)\Big)\textbf{H}_i^H\right)^{-1}\textbf{H}_i\left(\bar{\textbf{S}}^{(1)} - \bar{\textbf{S}}^{(0)}\right)\textbf{H}_i^H\right)\label{first_der},
\end{eqnarray}
\hrulefill
\begin{eqnarray}
\frac{\partial^2}{\partial \mu^2} \log\det\Big(\textbf{I} &+& \textbf{H}_i\Big(\bar{\textbf{S}}^{(0)} + \mu \left(\bar{\textbf{S}}^{(1)} - \bar{\textbf{S}}^{(0)}\right)\Big)\textbf{H}_i^H\Big)  = - \Tr\left(\textbf{A}^{-1}_i\textbf{H}_i\left(\bar{\textbf{S}}^{(1)} - \bar{\textbf{S}}^{(0)}\right)\textbf{H}_i^H\textbf{A}^{-1}_i\textbf{H}_i\left(\bar{\textbf{S}}^{(1)} - \bar{\textbf{S}}^{(0)}\right)\textbf{H}_i^H\right),\label{sec_d1}
\end{eqnarray}
\hrulefill
\end{figure*}

Note that as matrix $\textbf{G}_i$ is hermitian, the real operator is no longer needed since the trace of the product of two hermitian matrices is real. In order to prove condition $(\text{A}2)$, it suffices to show that for each linear cut in any direction, the surrogate function is a lower bound. Let $\bar{\textbf{S}} = \bar{\textbf{S}}^{(0)} + \mu \left(\bar{\textbf{S}}^{(1)} - \bar{\textbf{S}}^{(0)}\right)$, $\forall \mu \in [0,1]$. Then, it suffices to show \eqref{lin_ap}. Since the left hand side of \eqref{lin_ap} is concave with respect to $\mu$, a sufficient condition is that the second derivative of the left hand side of \eqref{lin_ap} must be lower than or equal to the second derivative of the right hand side of \eqref{lin_ap} for any $\mu\in[0,1]$ and any $\bar{\textbf{S}}^{(1)},\,\bar{\textbf{S}}^{(0)}\in\mathcal{S}_+^{n_T}$, thus, \eqref{tayl_b} must hold.

Let us compute the second derivative of the right hand side of \eqref{tayl_b}. The first derivative is given by \eqref{first_der} and the second derivative is given by \eqref{sec_d1}, where we have used the identity $\text{d}\textbf{X}^{-1} = - \textbf{X}^{-1}\text{d}\textbf{X}\textbf{X}^{-1}$ \cite{magnus} and matrix $\textbf{A}_i  \in\mathbb{C}^{n_{R_i}\times n_{R_i}}$ is defined as $\textbf{A}_i = \textbf{I} + \textbf{H}_i\Big(\bar{\textbf{S}}^{(0)} + \mu \left(\bar{\textbf{S}}^{(1)} - \bar{\textbf{S}}^{(0)}\right)\Big)\textbf{H}_i^H$.

We need to manipulate the previous expressions. To this end, let us define matrix $\textbf{P}_i = \textbf{H}^H_i\textbf{A}^{-1}_i\textbf{H}_i \in\mathbb{C}^{n_T\times n_T}$ and let us vectorize the result found in \eqref{sec_d1}:
\begin{eqnarray}
&\hspace{-5mm}\Tr&\Bigg(\textbf{P}_i\Big(\bar{\textbf{S}}^{(1)} - \bar{\textbf{S}}^{(0)}\Big)\textbf{P}_i\left(\bar{\textbf{S}}^{(1)} - \bar{\textbf{S}}^{(0)}\right)\Bigg) \nonumber \\
&\hspace{-12mm}=&\hspace{-5mm} \text{vec}\left(\left(\bar{\textbf{S}}^{(1)} - \bar{\textbf{S}}^{(0)}\right)^T\right)^T\hspace{-2mm}\left(\textbf{I}\otimes\textbf{P}^T_i\textbf{P}_i\right)\text{vec}\left(\bar{\textbf{S}}^{(1)} - \bar{\textbf{S}}^{(0)}\right)\label{vec4},
\end{eqnarray}
where we have used the following properties: $\Tr(\textbf{A}\textbf{B}) = \text{vec}(\textbf{A}^T)^T\text{vec}(\textbf{B})$, $\text{vec}(\textbf{A}\textbf{B})^T = \text{vec}(\textbf{A})^T(\textbf{I}\otimes\textbf{B})$, $\text{vec}(\textbf{A}\textbf{B}) = (\textbf{I}\otimes\textbf{A})\text{vec}(\textbf{B})$, and $(\textbf{A}\otimes\textbf{B})(\textbf{C}\otimes\textbf{D}) = (\textbf{A}\textbf{C})\otimes(\textbf{B}\textbf{D})$. Let us now vectorize the left hand side of \eqref{tayl_b}:
\begin{eqnarray}
&\hspace{-7mm}2&\hspace{-1mm}\Tr\Bigg(\Big(\bar{\textbf{S}} - \bar{\textbf{S}}^{(0)}\Big)^H\textbf{M}_i \Big(\bar{\textbf{S}} - \bar{\textbf{S}}^{(0)}\Big)\Bigg) \nonumber\\
&\hspace{-7mm}=&\hspace{0mm} 2\text{vec}\left(\left(\bar{\textbf{S}}^{(1)} - \bar{\textbf{S}}^{(0)}\right)^T\right)^T(\textbf{I}\otimes\textbf{M}_i)\text{vec}\left(\bar{\textbf{S}}^{(1)} - \bar{\textbf{S}}^{(0)}\right),\label{vec5}
\end{eqnarray}
where in \eqref{vec5} we have used the fact that $\bar{\textbf{S}}^{(1)} - \bar{\textbf{S}}^{(0)}$ is hermitian and $\Tr(\textbf{A}\textbf{B}\textbf{C}) = \text{vec}(\textbf{A}^T)^T(\textbf{I}\otimes\textbf{B})\text{vec}(\textbf{C})$. Finally, we end up with the relation from forcing that \eqref{vec5} must be lower than or equal to \eqref{vec4}. This relation can be expressed as given by \eqref{eq_53}. 
\begin{figure*}
\begin{equation}
2\text{vec}\left(\left(\bar{\textbf{S}}^{(1)} - \bar{\textbf{S}}^{(0)}\right)^T\right)^T\left[(\textbf{I}\otimes\textbf{M}_i) + \frac{1}{2}\left(\textbf{I}\otimes\textbf{P}^T_i\textbf{P}_i\right)\right]\text{vec}\left(\bar{\textbf{S}}^{(1)} - \bar{\textbf{S}}^{(0)}\right) \le 0.\label{eq_53}
\end{equation}
\hrulefill
\end{figure*}
A sufficient condition for expression \eqref{eq_53} is:
\begin{equation}
(\textbf{I}\otimes\textbf{M}_i) + \frac{1}{2}\left(\textbf{I}\otimes\textbf{P}^T_i\textbf{P}_i\right) = \textbf{I} \otimes \left(\textbf{M}_i+ \frac{1}{2}\textbf{P}^T_i\textbf{P}_i\right)\preceq 0,
\end{equation}
which means that 
\begin{equation}
\textbf{M}_i+ \frac{1}{2}\textbf{P}^T_i\textbf{P}_i\preceq 0.
\end{equation}
Now, if we set $\textbf{M}_i = \alpha\textbf{I}$ (note that this is a particular simple solution), we have that
\begin{equation}
\alpha \le - \frac{1}{2} \lambda_{\max}\left(\textbf{P}^T_i\textbf{P}_i\right),
\end{equation}
where $\lambda_{\max}(\textbf{X})$ is the maximum eigenvalue of matrix $\textbf{X}$. Now, let us introduce the following result:\\

\begin{theorem}[\cite{bwang:92}]
Let $\textbf{A}$, $\textbf{B} \in \mathbb{C}^{n\times n}$, assume that $\textbf{A}$ is positive definite, and assume that $\textbf{B}$ is positive definite. Let $\lambda_i(\textbf{A})$ be the $i$-th eigenvalue of matrix $\textbf{A}$ such that $\lambda_1(\textbf{A}) \ge \lambda_2(\textbf{A})\ge \dots \ge \lambda_n(\textbf{A})$. Then, for all $i,j,k\in\{1,\dots,n\}$ such that $j+k\le i+1$,
\begin{equation}
\lambda_i(\textbf{A}\textbf{B}) \le \lambda_j(\textbf{A})\lambda_k(\textbf{B}).
\end{equation}
In particular, for all $i=1,\dots,n$,
\begin{equation}
\lambda_i(\textbf{A})\lambda_n(\textbf{B}) \le \lambda_i(\textbf{A}\textbf{B}) \le \lambda_i(\textbf{A})\lambda_1(\textbf{B}).
\end{equation}
\end{theorem}

Thanks to the previous result, $\alpha \le - \frac{1}{2} \lambda^2_{\max}\left(\textbf{P}_i\right)$. Now, let the singular value decomposition of $\textbf{H}_i$ be $\textbf{H}_i = \textbf{U}_i\boldsymbol\Sigma_i\textbf{V}_i^H$. From this, we can upper bound $\lambda_{\max}\left(\textbf{P}_i\right) = \lambda_{\max}\left(\textbf{H}_i^H\textbf{A}_i^{-1}\textbf{H}_i\right) = \lambda_{\max}\left(\boldsymbol\Sigma_i\textbf{V}_i^H\textbf{A}_i^{-1}\textbf{V}_i\boldsymbol\Sigma_i\right) \le \sigma_{\max}^2(\textbf{H}_i)\lambda^{-1}_{\min}(\textbf{A}_i)$, where $\sigma_{\max}(\textbf{X})$ is the maximum singular value of matrix $\textbf{X}$. Because matrix $\textbf{A}$ is positive definite with $\lambda_{\min}(\textbf{A}_i) \ge 1$, we can conclude that
\begin{equation}
\alpha \le -\frac{1}{2} \sigma_{\max}^4(\textbf{H}_i), 
\end{equation}
and thus, a possible matrix $\textbf{M}_i$ satisfying conditions $(\text{A}1)-(\text{A}4)$ is finally
\begin{equation}
\textbf{M}_i = -\frac{1}{2} \sigma_{\max}^4(\textbf{H}_i) \textbf{I} = -\frac{1}{2} \lambda_{\max}^2(\textbf{H}^H_i\textbf{H}_i) \textbf{I}.
\end{equation}

\section{Proof of Proposition \ref{prop_p1}}
\label{app3}
Let us start by vectorizing the surrogate function in \eqref{surg_a}:
\begin{eqnarray}
\hat{R}_i(\textbf{S},\textbf{S}^{(0)})  &=& \hat{s}_i(\bar{\textbf{S}}, \bar{\textbf{S}}^{(0)}) - \hat{g}_i(\boldsymbol\Omega_i(\textbf{S}_{-i}),\boldsymbol\Omega_i^{(0)})\nonumber\\ 
&=&\text{vec}\left(\bar{\textbf{S}}^T\right)^T\left(\textbf{I}\otimes\textbf{M}_i\right)\text{vec}\left(\bar{\textbf{S}}\right)  + \textbf{e}^T_i\text{vec}\left(\bar{\textbf{S}}\right) \nonumber\\
&&+ {\textbf{r}}^T_i\text{vec}\left({\textbf{S}}_i\right) + \kappa_2,\label{59}
\end{eqnarray}
where $\textbf{e}_i = \text{vec}\left(\textbf{E}_i^T\right) \in\mathbb{C}^{n_Tn_T\times 1}$, ${\textbf{r}}_i = \text{vec}\left(\textbf{R}_i^T\right) \in\mathbb{C}^{n_Tn_T\times 1}$, and $\kappa_2$ contains some constant terms that do not depend on $\{\textbf{S}_i\}$. Let ${\textbf{s}} = \left[\text{vec}(\textbf{S}_1)^T \text{vec}(\textbf{S}_2)^T \dots \text{vec}(\textbf{S}_{|\mathcal{U}_I|})^T\right]^T \in\mathbb{C}^{n_Tn_T|\mathcal{U}_I|\times 1}$. Note that $\text{vec}\left(\bar{\textbf{S}}\right) = \textbf{T}{\textbf{s}}$, where $\textbf{T} \in\mathbb{C}^{n_Tn_T\times n_Tn_T|\mathcal{U}_I|}$ is composed of $|\mathcal{U}_I|$ identity matrices of size $n_Tn_T\times n_Tn_T$, i.e.,  $\textbf{T}= \left[\textbf{I}\,\, \textbf{I}\,\,\dots \,\,\textbf{I}\right]$. Now, we can rewrite \eqref{59} as (omitting the constant terms)
\begin{eqnarray}
\hspace{-5mm}\hat{R}_i(\textbf{S},\textbf{S}^{(0)}) \hspace{-1mm}&=& \hspace{-1mm}{\textbf{s}}^H\textbf{T}^H\left(\textbf{I}\otimes\textbf{M}_i\right)\textbf{T}{\textbf{s}} + {\textbf{e}}_i^T\textbf{T}{\textbf{s}}  + {\textbf{r}}_i^T\text{vec}\left({\textbf{S}}_i\right).
\end{eqnarray}

We know proceed to formulate the objective function (denoted by $\bar{f}_0(\textbf{S},\textbf{S}^{(0)})$ of problem \eqref{op:wet1} but substituting the bound that we just computed and considering the proximal term. If we incorporate all the terms (but omitting the constant ones) we have
\begin{eqnarray}
&\bar{f}_0&(\textbf{S},\textbf{S}^{(0)}) = \nonumber\\
&&\sum_{i\in\mathcal{U}_I} \omega_i\Bigg({\textbf{s}}^H\textbf{T}^H\left(\textbf{I}\otimes\textbf{M}_i\right)\textbf{T}{\textbf{s}}  + {\textbf{e}}_i^T\textbf{T}{\textbf{s}}  + {\textbf{r}}_i^T\text{vec}\left({\textbf{S}}_i\right)\Bigg) \nonumber\\
&& - \rho \left\|\textbf{S}_i- \textbf{S}^{(0)}_i\right\|_F^2\\
&=& {\textbf{s}}^H\textbf{T}^H\tilde{\textbf{M}}\textbf{T}{\textbf{s}}   + {\tilde{\textbf{e}}}^T\textbf{T}{\textbf{s}} + \hat{\textbf{r}}^T{\textbf{s}} - \rho {\textbf{s}}^H{\textbf{s}} + \rho {\textbf{s}}^{(0),H}{\textbf{s}} \nonumber\\
&&+ \rho{\textbf{s}}^H{\textbf{s}}^{(0)} - \rho{\textbf{s}}^{(0),H}{\textbf{s}}^{(0)},
\end{eqnarray}
where $\tilde{\textbf{M}} = \sum_{i\in\mathcal{U}_I} \omega_i \left(\textbf{I}\otimes\textbf{M}_i\right) \in\mathbb{C}^{n_Tn_T\times n_Tn_T}$, ${\tilde{\textbf{e}}} =  \sum_{i\in\mathcal{U}_I} \omega_i {\textbf{e}}_i$, $\hat{\textbf{r}} = \left[{\textbf{r}}_1^T \,{\textbf{r}}_2^T\,\dots\, {\textbf{r}}_{|\mathcal{U}_I|}^T\right]^T \in\mathbb{C}^{n_Tn_T|\mathcal{U}_I|\times 1}$, and ${\textbf{s}}^{(0)} = \left[\text{vec}(\textbf{S}^{(0)}_1)^T \text{vec}(\textbf{S}^{(0)}_2)^T \dots \text{vec}(\textbf{S}^{(0)}_{|\mathcal{U}_I|})^T\right]^T\in\mathbb{C}^{n_Tn_T|\mathcal{U}_I|\times 1}$. Now taking into account that the objective function $\bar{f}_0(\textbf{S},\textbf{S}^{(0)})$ must be real and combining terms (omitting terms that do not depend on ${\textbf{s}}$) we obtain
\begin{eqnarray}
\bar{f}_0(\textbf{S},\textbf{S}^{(0)}) =  {\textbf{s}}^H\textbf{C}{\textbf{s}} + \textbf{b}^T{\textbf{s}} +{\textbf{s}}^H\textbf{b}^*, \label{ob_fun}
\end{eqnarray}
where $\textbf{b}^T =  \frac{1}{2}{\tilde{\textbf{e}}}^T\textbf{T} + \frac{1}{2}\hat{\textbf{r}}^T  +\rho{\textbf{s}}^{(0),H}\in\mathbb{C}^{1\times n_Tn_T|\mathcal{U}_I|}$ and matrix $\textbf{C}$ is $\textbf{C} = \textbf{T}^H\tilde{\textbf{M}}\textbf{T} - \rho\textbf{I}\in\mathbb{C}^{n_Tn_T|\mathcal{U}_I|\times n_Tn_T|\mathcal{U}_I|}$. For convenient purposes, let us change the sign of $\bar{f}_0(\textbf{S},\textbf{S}^{(0)})$ such that $\bar{\bar{f}}_0(\textbf{S},\textbf{S}^{(0)}) = - \bar{f}_0(\textbf{S},\textbf{S}^{(0)}) =  {\textbf{s}}^H\tilde{\textbf{C}}{\textbf{s}} - \textbf{b}^T{\textbf{s}} - {\textbf{s}}^H\textbf{b}^*$, where $\tilde{\textbf{C}} = - \textbf{C} \succeq 0$. Finally, we can equivalently rewrite the objective function as the following expression (with this new reformulation, the objective is to minimize $\bar{\bar{f}}_0(\textbf{S},\textbf{S}^{(0)})$ instead of maximizing it):
\begin{equation}
\bar{\bar{f}}_0(\textbf{S},\textbf{S}^{(0)}) = \|\tilde{\textbf{C}}^{\frac{1}{2}}{\textbf{s}} - \textbf{c}\|^2_2,
\end{equation}
where 
\begin{equation}
\textbf{c} = \tilde{\textbf{C}}^{-\frac{1}{2}}\textbf{b}^*\in\mathbb{C}^{n_Tn_T|\mathcal{U}_I|\times 1}.\label{eq_c}
\end{equation}
Note that the term $\textbf{c}^H\textbf{c}$ does not affect the optimum value of the optimization variables as this term does not depend on ${\textbf{s}}$. Now, we can reformulate the optimization problem presented in \eqref{op:wet1} as 
\begin{alignat}{2}
\mathop{\text{minimize}}_{\{\textbf{S}_i\}, \,{\textbf{s}}}& \quad  \|\tilde{\textbf{C}}^{\frac{1}{2}}{\textbf{s}} - \textbf{c}\|^2_2\label{op:wetN1} \\
\textrm{subject to} 
& \quad C1: \textbf{T}_i{\textbf{s}} = \text{vec}\left(\textbf{S}_i\right),&&\quad \forall i \in\mathcal{U}_I\nonumber\\
&  \quad C2: \textbf{S} \in\mathcal{S}_1 \nonumber,
\end{alignat}
where $\textbf{T}_i = [\underbrace{\textbf{0}, \textbf{0}, \dots, \textbf{0}}_{i-1}, \textbf{I}, \textbf{0},\dots,\textbf{0}]\in\mathbb{R}^{n_Tn_T\times n_Tn_T|\mathcal{U}_I|}$ is composed of zero matrices of dimension $n_Tn_T\times n_Tn_T$ with an identity matrix at the $i$-th position. Problem \eqref{op:wetN1} can be further reformulated as
\begin{alignat}{2}
\mathop{\text{minimize}}_{\{\textbf{S}_i\}, \,{\textbf{s}},\,t}& \quad  t \label{op:wetN2} \\
\textrm{subject to} 
& \quad  C1: \|\tilde{\textbf{C}}^{\frac{1}{2}}{\textbf{s}} - \textbf{c}\|_2 \le t \nonumber\\
& \quad C2: \textbf{T}_i{\textbf{s}} = \text{vec}\left(\textbf{S}_i\right),&&\quad \forall i \in\mathcal{U}_I\nonumber\\
&  \quad C3: \textbf{S} \in\mathcal{S}_1 \nonumber,
\end{alignat}
and, finally, as the following SDP optimization problem
\begin{alignat}{2}
\mathop{\text{minimize}}_{\{\textbf{S}_i\}, \,{\textbf{s}},\,t}& \quad  t \label{op:wetN3} \\
\textrm{subject to} 
& \quad C1: \left[ \begin{array}{cc}
     {t}\textbf{I} & \tilde{\textbf{C}}^{\frac{1}{2}}{\textbf{s}} - \textbf{c} \\ \left(\tilde{\textbf{C}}^{\frac{1}{2}}{\textbf{s}} - \textbf{c}\right)^H & 1 \end{array} \right] \succeq 0 \nonumber \\
& \quad C2: \textbf{T}_i{\textbf{s}} = \text{vec}\left(\textbf{S}_i\right),&&\quad \forall i \in\mathcal{U}_I\nonumber\\
&  \quad C3: \textbf{S} \in\mathcal{S}_1 \nonumber.
\end{alignat}

\section{Proof of Proposition \ref{prop_sf2}}
\label{app4}
The proposed quadratic surrogate function of $s_i(\textbf{S})$ has the following form:
\begin{eqnarray}
\hat{s}_i(\textbf{S}&\hspace{-4mm},&\hspace{-4mm}\textbf{S}^{(0)}) \triangleq \log\det\left(\textbf{I} + \textbf{H}_i\sum_{k\in\mathcal{U}_I}\textbf{S}_k^{(0)}\textbf{H}_i^H\right) \label{sur_ap2b}\\
&&+ \sum_{\ell\in\mathcal{U}_I}\text{Re}\Big\{\Tr\left(\textbf{G}_{\ell i}\left(\textbf{S}_\ell - \textbf{S}_\ell^{(0)}\right)\right)\Big\} \nonumber\\
& & + \sum_{\ell\in\mathcal{U}_I}\Tr\left(\left(\textbf{S}_\ell - \textbf{S}_\ell^{(0)}\right)^H\textbf{M}_{\ell i} \left(\textbf{S}_\ell - \textbf{S}_\ell^{(0)}\right)\right)\nonumber \\
&\le&   \log\det\left(\textbf{I} + \textbf{H}_i\sum_{k\in\mathcal{U}_I}\textbf{S}_k\textbf{H}_i^H\right), \quad \forall  \textbf{S}_\ell,\, \textbf{S}_\ell^{(0)} \in\mathcal{S}^{n_T}_+,\nonumber
\end{eqnarray}
where matrices $\textbf{G}_i \in\mathbb{C}^{n_T\times n_T}$ and $\textbf{M}_i \in\mathbb{C}^{n_T\times n_T}$ need to be found such that conditions $(\text{A}1)$ through $(\text{A}4)$ are satisfied. Note that $(\text{A}1)$ and $(\text{A}4)$ are already satisfied. Only $(\text{A}2)$ and $(\text{A}3)$ must be ensured. 
\begin{figure*}
\setcounter{equation}{66}
\begin{eqnarray}
\Tr\Bigg(\textbf{H}^H_i\Bigg(\textbf{I} + \textbf{H}_i\sum_{k\in\mathcal{U}_I}\textbf{S}_k^{(0)}\textbf{H}_i^H\Bigg)^{-1}&&\hspace{-6mm}\textbf{H}_i\left(\sum_{\ell\in\mathcal{U}_I} \left({\textbf{S}}_\ell^{(1)} - {\textbf{S}}_\ell^{(0)}\right)\right)\Bigg) \nonumber\\
&=& \sum_{\ell\in\mathcal{U}_I}\Tr\left(\textbf{H}^H_i\left(\textbf{I} + \textbf{H}_i\sum_{k\in\mathcal{U}_I}\textbf{S}_k^{(0)}\textbf{H}_i^H\right)^{-1}\textbf{H}_i\left({\textbf{S}}_\ell^{(1)} - {\textbf{S}}_\ell^{(0)}\right)\right)\label{eq68}
\end{eqnarray}
\hrulefill
\setcounter{equation}{69}
\begin{equation}
\text{vec}\left(\left(\sum_{\ell\in\mathcal{U}_I}\left({\textbf{S}}^{(1)}_\ell - {\textbf{S}}^{(0)}_\ell\right)\right)^T\right)^T\left(\textbf{I}\otimes\textbf{P}^T_i\textbf{P}_i\right)\text{vec}\left(\sum_{\ell\in\mathcal{U}_I}\left({\textbf{S}}^{(1)}_\ell - {\textbf{S}}^{(0)}_\ell\right)\right), \label{sec_rhs}
\end{equation}
\hrulefill
\end{figure*}
Let us start with condition $(\text{A}3)$. Let ${\textbf{S}}_\ell^{(0)}$, ${\textbf{S}}_\ell^{(1)} \in\mathcal{S}^{n_T}_+$, $\forall \ell$. Then, the directional derivative of the surrogate function $\hat{s}_i(\textbf{S},\textbf{S}^{(0)})$ in \eqref{sur_ap2b} at $\textbf{S}_\ell^{(0)}$ with direction $\textbf{S}_\ell^{(1)} - {\textbf{S}}_\ell^{(0)}$ is given by
\setcounter{equation}{65}
\begin{equation}
\sum_{\ell \in\mathcal{U}_I} \text{Re}\Big\{\Tr\left(\textbf{G}_{\ell i}\left({\textbf{S}}_\ell^{(1)} - {\textbf{S}}_\ell^{(0)}\right)\right)\Big\},\label{eq66}
\end{equation}
and the directional derivative of the right hand side of \eqref{sur_ap2b} at $\textbf{S}_\ell^{(0)}$ with direction $\textbf{S}_\ell^{(1)} - {\textbf{S}}_\ell^{(0)}$ is given by \eqref{eq68}. From \eqref{eq66} and \eqref{eq68}, we identify the matrices $\textbf{G}_{\ell i}$ as
\setcounter{equation}{67}
\begin{equation}
\textbf{G}_{\ell i} = \textbf{H}^H_i\left(\textbf{I} + \textbf{H}_i\sum_{k\in\mathcal{U}_I}\textbf{S}_k^{(0)}\textbf{H}_i^H\right)^{-1}\textbf{H}_i, \quad \textbf{G}_{\ell i} = \textbf{G}^H_{\ell i},
\end{equation}
where we find that all matrices $\textbf{G}_{\ell i}$ for a given user $i$ can be the same, $\textbf{G}_i = \textbf{G}_{\ell i}$ (i.e., they do not depend on $\ell$). 

Now, we seek to find matrices $\{\textbf{M}_{\ell i}\}$ based on condition $(\text{A}2)$. To this end, we follow the same procedure presented before. We make linear cuts in each possible direction and apply the condition over the second derivative (see \eqref{tayl_b}). The second derivative of the left hand side of \eqref{sur_ap2b} is given by
\begin{eqnarray}
&\hspace{-3mm}&2\sum_{\ell\in\mathcal{U}_I}\Tr\left(\left(\textbf{S}_\ell^{(1)} - \textbf{S}_\ell^{(0)}\right)^H\textbf{M}_{\ell i} \left(\textbf{S}_\ell^{(1)} - \textbf{S}_\ell^{(0)}\right)\right) = \\
&\hspace{-3mm}&2\sum_{\ell\in\mathcal{U}_I}\text{vec}\left(\left(\textbf{S}_\ell^{(1)} - \textbf{S}_\ell^{(0)}\right)^T\right)^T\left(\textbf{I}\otimes\textbf{M}_{\ell i}\right)\text{vec}\left(\textbf{S}_\ell^{(1)} - \textbf{S}_\ell^{(0)}\right),\nonumber
\end{eqnarray}
and the second derivative of the right hand side is given by \eqref{sec_rhs}, where $\textbf{P}_i = \textbf{H}_i^H\left(\textbf{I} + \textbf{H}_i\Big(\sum_{\ell\in\mathcal{U}_I}\left({\textbf{S}}^{(0)}_\ell + \mu \left({\textbf{S}}^{(1)}_\ell - {\textbf{S}}^{(0)}_\ell\right)\right)\Big)\textbf{H}_i^H\right)^{-1}\textbf{H}_i$, being constant $\mu \in [0,1]$. Now, let ${\textbf{s}} = \left[\text{vec}\left(\textbf{S}_1^{(1)}-\textbf{S}_1^{(0)}\right)^T  \cdots \text{vec}\left(\textbf{S}_{|\mathcal{U}_I|}^{(1)}-\textbf{S}_{|\mathcal{U}_I|}^{(0)}\right)^T\right]^T$ and 
let us introduce the following block diagonal matrix
\begin{equation}
\setcounter{equation}{71}
\tilde{\textbf{M}}_{i} = \left[\begin{array}{cccc}
\textbf{I}\otimes\textbf{M}_{1i} &\textbf{0}  &\hdots &\textbf{0}  \\
\textbf{0} &\textbf{I}\otimes\textbf{M}_{2i}& & \vdots  \\
\vdots & &\ddots & \textbf{0} \\
\textbf{0}  &\hdots & \textbf{0}  &\textbf{I}\otimes\textbf{M}_{|\mathcal{U}_I|i}\\
\end{array}\right].
\end{equation}
Then we have that the following condition should be fulfilled:
\begin{equation}
2{\textbf{s}}^H\tilde{\textbf{M}}_i{\textbf{s}} + {\textbf{s}}^H \textbf{T}^H \left(\textbf{I}\otimes\textbf{P}^T_i\textbf{P}_i\right)\textbf{T} {\textbf{s}} \le 0,
\end{equation}
which means that 
\begin{equation}
\tilde{\textbf{M}}_i + \frac{1}{2}\textbf{T}^H\left(\textbf{I}\otimes\textbf{P}^T_i\textbf{P}_i\right)\textbf{T}\preceq 0.\label{eq75}
\end{equation}
Note that the particular structure of matrix $\textbf{T}^H\left(\textbf{I}\otimes\textbf{P}^T_i\textbf{P}_i\right)\textbf{T}$ is given by
\begin{equation}
\textbf{T}^H\left(\textbf{I}\otimes\textbf{P}^T_i\textbf{P}_i\right)\textbf{T} = \left[\begin{array}{cccc}
\textbf{I}\otimes\textbf{P}^T_i\textbf{P}_i   &\hdots &\textbf{I}\otimes\textbf{P}^T_i\textbf{P}_i \\
\textbf{I}\otimes\textbf{P}^T_i\textbf{P}_i  &   \\
\vdots & \ddots & \vdots\\
\textbf{I}\otimes\textbf{P}^T_i\textbf{P}_i  &\hdots  &\textbf{I}\otimes\textbf{P}^T_i\textbf{P}_i\\
\end{array}\right],
\end{equation}

From the previous conditions we can see that all matrices $\textbf{M}_{\ell i}$ will be the same for user $i$, i.e., $\textbf{M}_{\ell i} = \textbf{M}_i, \,\forall \ell$. Now if we choose the particular structure $\textbf{M}_i = \alpha_i\textbf{I}$, then condition \eqref{eq75} is equivalent to
\begin{equation}
\alpha_i\textbf{I} + \frac{1}{2}\textbf{T}^H\left(\textbf{I}\otimes\textbf{P}^T_i\textbf{P}_i\right)\textbf{T}\preceq 0.\label{77}
\end{equation}
Now, condition \eqref{77} is equivalent to
\begin{eqnarray}
\alpha_i\textbf{g}^H\textbf{g} &\le& - \frac{1}{2}\textbf{g}^H\textbf{T}^H\left(\textbf{I}\otimes\textbf{P}^T_i\textbf{P}_i\right)\textbf{T}\textbf{g}, \quad \forall \textbf{g} \quad \Longrightarrow\\
\alpha_i\textbf{g}^H\textbf{g} &\le& - \frac{1}{2} \|\textbf{T}\textbf{g}\|_2^2\lambda_{\max}\left(\textbf{I}\otimes\textbf{P}^T_i\textbf{P}_i\right), \quad \forall \textbf{g}\quad  \Longrightarrow\\
\alpha_i\textbf{g}^H\textbf{g} &\le& - \frac{1}{2} \|\textbf{T}\textbf{g}\|_2^2\lambda_{\max}\left(\textbf{P}^T_i\textbf{P}_i\right), \quad \forall \textbf{g}.\label{80}
\end{eqnarray}
Now, the term $\|\textbf{T}\textbf{g}\|_2^2$ can be further simplified. Based on the structure of matrix $\textbf{T}$, we have that 
\begin{eqnarray}
&\hspace{-3mm}\|&\hspace{-2mm}\textbf{T}\textbf{g}\|_2^2 =\nonumber\\ 
&&\hspace{-1mm}\sum_{i=1}^{n_Tn_T} |\textbf{g}_{i} + \textbf{g}_{i+n_Tn_T+1} + \ldots + \textbf{g}_{i+n_Tn_T(|\mathcal{U}_I|-1)+1} |^2\\
&\hspace{-4mm}\le&\hspace{-1mm}  \sum_{i=1}^{n_Tn_T} ||\mathcal{U}_I| \max\{\textbf{g}_{i}, \ldots, \textbf{g}_{i+n_Tn_T(|\mathcal{U}_I|-1)+1}\}|^2\\
&\hspace{-4mm}\le& \sum_{i=1}^{n_Tn_T} |\mathcal{U}_I|^2 \left(|\textbf{g}_{i}|^2 + \ldots + |\textbf{g}_{i+n_Tn_T(|\mathcal{U}_I|-1)+1}|^2 \right)\\
&\hspace{-4mm}=& |\mathcal{U}_I|^2 \sum_{i=1}^{n_Tn_T|\mathcal{U}_I|} |\textbf{g}_i|^2 = |\mathcal{U}_I|^2\|\textbf{g}\|_2^2.
\end{eqnarray}
Thus, a sufficient condition to fulfill \eqref{80} is 
\begin{equation}
\alpha_i\|\textbf{g}\|_2^2 \le- \frac{1}{2} |\mathcal{U}_I|^2\|\textbf{g}\|_2^2\lambda_{\max}\left(\textbf{P}^T_i\textbf{P}_i\right), \quad \forall \textbf{g},
\end{equation}
and, finally,
\begin{equation}
\alpha_i \le - \frac{1}{2} |\mathcal{U}_I|^2 \lambda_{\max}\left(\textbf{P}^T_i\textbf{P}_i\right) \le - \frac{1}{2} |\mathcal{U}_I|^2\lambda_{\max}^2(\textbf{H}^H_i\textbf{H}_i). 
\end{equation}
Hence, a possible matrix $\textbf{M}_i$ satisfying assumptions $(\text{A}1)-(\text{A}4)$ is, finally,
\begin{equation}
\textbf{M}_i = -\frac{1}{2} |\mathcal{U}_I|^2\lambda_{\max}^2(\textbf{H}^H_i\textbf{H}_i) \textbf{I}.
\end{equation}


\bibliographystyle{IEEEtran}
\bibliography{referencias}

\end{document}